\documentclass[11pt]{article}

\usepackage[numbers]{natbib}
\usepackage{lipsum}
\usepackage{balance}
\usepackage{blindtext}
\usepackage{hyperref}
\usepackage{xcolor}
\usepackage{amsfonts}
\usepackage{multirow}
\usepackage{makecell}
\usepackage{tabu}
\usepackage{graphicx}
\usepackage{enumitem}
\usepackage{algorithm}
\usepackage{algpseudocode}
\usepackage[english]{babel}
\usepackage{amsthm}
\usepackage{amsmath}
\usepackage{cleveref}
\usepackage{array}
\usepackage{booktabs}
\usepackage{soul}
\usepackage{authblk}
\usepackage{amsmath}
\usepackage{adjustbox}
\usepackage{tabularx}
\usepackage{float}
\usepackage[bibliography=common]{apxproof}

\newcommand{\BibTeX}{\rm B\kern-.05em{\sc i\kern-.025em b}\kern-.08em\TeX}
\newtheoremrep{theorem}{Theorem}[section]
\newtheoremrep{proposition}[theorem]{Proposition}
\newtheoremrep{lemma}[theorem]{Lemma}

\newtheorem{Corollary}[theorem]{Corollary}
\theoremstyle{definition}
\newtheorem{Definition}[theorem]{Definition}
\newtheorem{Example}[theorem]{Example}
\newtheorem{Remark}[theorem]{Remark}
\newtheorem{Open}[theorem]{Open Question}

\newcommand{\no}[1]{\textcolor{black}{#1}}
\newcommand{\noga}[1]{\no{(Noga says: #1)}}
\newcommand{\fullversion}[1]{}

\newcommand{\price}{\mathbf{p}}

\makeatletter
\newcounter{phase}[algorithm]
\newlength{\phaserulewidth}

\makeatother

\makeatletter
\title{Fair Allocation of Improvements: 
When Old Endowments Shape New Assignments}

\title{Fair Allocation of Improvements:\\ When Old Endowments Shape New Assignments}
\author[1]{Noga Klein Elmalem}
\author[2]{Rica Gonen}
\author[3]{Erel Segal-Halevi}

\affil[1]{The Open University of Israel\\noga486@gmail.com}
\affil[2]{The Open University of Israel\\ricagonen@gmail.com}
\affil[3]{Ariel University, Israel\\erelsgl@gmail.com}

\begin{document}
\maketitle
\begin{abstract}
This work is motivated by a common urban renewal process called Reconstruct and Divide. It involves the demolition of old buildings and the construction of new ones. Original homeowners are compensated with upgraded apartments, while surplus units are sold for profit, so theoretically it is a win-win project for all parties involved.
However, many Reconstruct and Divide projects are withheld or delayed due to disagreements over the assignment of new apartments, claiming they are not fair. The goal of this research is to develop algorithms for envy-free assignment of the new apartments, possibly using monetary payments to reduce envy.

In contrast to previous works on envy-free assignment,  in our setting the envy depends also on the value of the old apartments, as people with more valuable old apartments expect to get more valuable new apartments. This presents two challenges.

First, in some cases, no assignment and payment-vector satisfy the common fairness notions of envy-freeness and proportionality.
Hence, we focus on minimizing the envy and the disproportionality
(the distance between an agent's value and their proportional share).
We present a strongly polynomial-time algorithm that, for a given assignment, finds a payment vector that minimizes the maximum pairwise-envy.
We also present a strongly polynomial-time algorithm that computes an assignment and payment-vector that together minimize the maximum disproportionality.

Second, directly asking the agents for their subjective valuations for their old apartments is infeasible, as it is a dominant strategy for them to report very high values for their old apartments. We introduce a novel  method to elicit agents' valuations indirectly. Using this method, we identify conditions under which our Minimum Disproportionality algorithm is risk-averse truthful.
\end{abstract}

\section{Introduction}
Urban renewal aims to transform aging residential infrastructure into modern, safe, and more spacious housing while simultaneously increasing urban density. 

\emph{Reconstruct and Divide}
 is a common urban renewal process. It involves demolishing an old building and constructing a new one, with original homeowners receiving upgraded apartments as compensation. While the primary goals include enhancing urban housing availability and improving disaster resilience, the success of these projects is often impeded by disputes among stakeholders, particularly concerning the assignment of new apartments.
These disagreements commonly stem from perceptions of unfairness, as homeowners compare the value of their newly-assigned units to others. Such cases often get to court and lead to lengthy judicial processes, leading to delays or even project cancellations.

\begin{Example}
\label[example]{example:court}
In one court proceeding we found the following claims by one homeowner:
(1) Her new apartment is larger than her old apartment by 19 sqm, whereas the new apartments of other homeowners are larger by 23 sqm;  (2) her new apartment is 1.5 floors higher than her old apartment, whereas the new apartments of other homeowners are higher by 2 floors; (3) her old apartment was in a half-floor with no adjacent neighbors, whereas her new apartment has adjacent neighbors; (4) her old apartment was square-shaped whereas her new apartment is not, but other homeowners did get a square-shaped apartment.

These claims were rejected by the judges, as they claimed that fairness only requires the equal treatment of equals, whereas homeowners with different old apartments are not apriori equal.
However, such subjective feelings of envy might lead to homeowners disagreeing to enter the process in the first place, foregoing its advantages for all parties involved.
\end{Example}



The baseline for this research is the problem often dubbed 
\emph{Rental Harmony} \citep{edward1999rental}, where some $n$ agents should be assigned some $n$ items (rooms or apartments) and some monetary payments, such that no agent envies the bundle (item + payment) of another agent. An envy-free assignment is known to exist under very broad conditions \cite{haake2002bidding,sung2004competitive,velez2018equitable,segal2022generalized,sanchez2022envy,peters2022robust,arunachaleswaran2022fully,velez2023equitable,airiau2023fair,procaccia2025multi}, and existing computational solutions are widely applied in practice \cite{gal2017fairest}.

The main challenge in the Reconstruct and Divide setting is that, as illustrated by \Cref{example:court}, the envy of agents is determined not by the value of their new apartment alone, but by the \emph{improvement} of their new apartment over their old one.
Even before going into details of how this improvement is computed, we can already see that an envy-free assignment might not exist.
\begin{Example}
\label[example]{example:no-ef} 
Suppose each agent believes that his old apartment is more valuable than all other old apartment, but the new apartments are all identical. In this case, regardless of the assignment, each agent will feel that the other agents got a better improvement. Payments might reduce the envy of some agents but increase the envy of others; no payment vector can totally eliminate the subjective envy (See \Cref{example:no-efable} for a formal proof).
\end{Example}
A common solution to this issue is to hire an appraiser to compute the market value of old apartments. However, as illustrated by \Cref{example:court}, agents' valuations are subjective and may differ from the market values.
\no{A central challenge in this work is to identify and formalize fairness and strategic guarantees that are both meaningful and attainable in the Reconstruct and Divide setting.}

As it might be impossible to completely eliminate envy, we focus on minimizing the envy of the most envious agent
(see \Cref{section:max-envy,section:proportionality} for the formal definition).

\emph{Proportionality} is another prominent fairness notion. In the standard fair division setting (with no old allocation), proportionality means that each agent should receive a bundle worth (subjectively) at least $1/n$ of the total, where the "total" is simply the sum of (subjective) values of all items. In our setting, as the allocation is of improvements, it makes sense to define the "total" as the sum of (subjective) improvements of all items.
For example, if you value all old apartments together at $100$ and all new apartments together at $150$, then proportionality guarantees that your own subjective value will improve by at least $50/n$.

Proportionality is weaker than envy-freeness: every envy-free allocation is proportional, but not the other way around. This holds both in the standard setting and in our setting (see \Cref{section:proportionality}).
Still, a proportional allocation might not exist.
This can be seen in \Cref{example:no-ef}: each agent believes that his improvement is smaller than the improvement of any other agent; hence, his improvement is necessarily smaller than $1/n$ the sum of improvements (see \Cref{section:proportionality} for a formal proof).
Similarly to the case of envy-freeness, we aim to minimize the maximum ``dis-proportionality'' --- the difference between each agent's proportional share to the agent's improvement.

\subsection{Our Results}
We assume throughout the paper that (a) all agents have quasilinear utilities; (b) all agents evaluate their improvements by computing the \emph{difference} between the values of their new apartments and their old apartments%
\footnote{
\label{footnote:ratio-model}
We considered an alternative model in which agents compute \emph{ratios} rather than differences, but could not yet get any interesting results for this model. See \Cref{section: Ratio Model} for details.
}.


In \Cref{section:max-envy} we study envy-freeness (EF). We provide a necessary and sufficient condition for an assignment to be ``EF-able'' (can be made EF using payments). 
Since an EF-able assignment might not exist, 
we aim to identify, for any given assignment, a payment-vector that minimizes the maximum envy over all agents. We prove that this minimum equals the \emph{maximum average cycle cost} in the envy-graph corresponding to the assignment. Hence, we can use existing algorithms to minimize the maximum envy for a given assignment in strongly-polynomial time%
\footnote{Strong polynomial time is a well-known concept in computer science, applicable to problems with numbers. Whereas a polynomial-time algorithm depends on the bit-length of input numbers, a strongly polynomial-time algorithm depends only on the number of input numbers. Strongly polynomial time is better as it avoids slowdowns from large or precise numbers.}.
However, we do not yet have a polynomial-time algorithm for finding an assignment that minimizes the maximum envy. Whether such an algorithm exists or not remains an intriguing open problem.

In \Cref{section:proportionality} we  study proportionality (PROP). We provide a necessary and sufficient condition for an assignment to be PROP-able. 
We show that it is possible to find, in polynomial time, an assignment and a payment-vector that minimize the maximum dis-proportionality over all agents; we call it the \emph{Minimum Disproportionality} mechanism. \no{Notably, this mechanism is based on a simple utilitarian-maximization algorithm originally developed for computing envy-freeable allocations.
We show that it can be repurposed to minimize disproportionality in the more complex setting with endowments.}

As the PROP model is more computationally tractable, we focus on this model and discuss, in \Cref{sec: manipulations}, its strategic aspects.
It is well-known that, even in the basic Rental Harmony setting, no fair and budget-balanced algorithm is truthful. But in the Reconstruct and Divide setting the situation is much worse, as the agents have a strong incentive to report the maximum possible value to their old apartment. In fact, this manipulation is a dominant strategy --- it is a \emph{safe manipulation} (see \Cref{sec: manipulations} for definitions).
To alleviate this issue, we introduce a new way to elicit agents' valuations: instead of asking them to directly report the values of old and new apartments, we ask them to report the values of apartments' \emph{characteristics}, such as apartment orientation, high floor etc. The reported characteristic may appear both in old and in new apartments. We explore the conditions under which the Minimum Disproportionality mechanism has no safe manipulations. \no{To the best of our knowledge, characteristic-based elicitation has not previously been used to improve strategic properties of fair-division mechanisms.}

In \Cref{sec: exp}, we simulated a Reconstruct and Divide scenario using survey-based valuations from 45 homeowners, comparing outcomes under minimizing envy and minimizing disproportionality models. The results reveal the gap between the models and underscore the importance and the need for improved allocations in the minimizing envy 
 model. \no{Our experiments also provide a novel dataset of valuations over item characteristics and show that the Minimum Disproportionality mechanism often attains better-than-proportional outcomes.}

All omitted proofs can be found in the supplementary material.

\subsection{Related Work}

\no{Our work connects several lines of research, including fair assignment with monetary transfers, initial endowments, variations of the house assignment problem, and relaxations of truthfulness.
}

\no{
The problem of minimizing subsidies to achieve envy-freeness was introduced by \citet{halpern2019fair}, who showed that any assignment can be made envy-free using monetary transfers. \citet{brustle2020one} extended this idea by allowing the assignment itself to vary and by providing tight bounds on the required subsidies. More recently, \citet{kawase2024towards} improved the known subsidy bounds by showing that from any EF1 allocation one can compute in polynomial time an envy-free allocation with per-agent subsidy at most $n-1$ and total subsidy at most $n(n-1)/2$, with even better bounds for monotone valuations.
}

\no{
Goko et al. \cite{goko2024fair} study an algorithm for computing envy-free allocations with subsidies that is also \emph{truthful}, when agents have submodular binary valuations.
}

\no{
Wu et al. \cite{wu2023one,wu2024tree} study proportional allocations with subsidies. They show that under additive valuations a total subsidy of $n/4$ is sufficient and tight to guarantee proportionality, and provide improved bounds and rounding schemes for the weighted case.
}

\no{
Finally, \citet{dai2024weighted} investigate weighted envy-freeness in the house allocation setting, presenting polynomial-time algorithms and structural characterizations of weighted envy-free and weighted envy-freeable assignments, under a definition of weighted envy-freeness that differs from the one adopted in our work.
}

\no{
For a broader perspective, \citet{liu2024mixed} survey recent advances in fair division and highlight open problems in mixed settings involving goods, chores, divisible and indivisible resources, and subsidies.
}

\no{
For ease of reference, a comparative summary of the discussed works, including their settings, fairness notions, and main results, is provided in ~\Cref{tab:related_work} 
}

\begin{table}[H]
\centering
\renewcommand{\arraystretch}{1.05}
\setlength{\tabcolsep}{6pt}

\begin{tabularx}{\textwidth}{|X|X|X|X|}
\hline
\textbf{Reference} & \textbf{Setting} & \textbf{Fairness} & \textbf{Main result} \\
\midrule
~\citet{halpern2019fair} 
& Indivisible goods 
& Envy-freeness
& Subsidy minimization; polynomial time algorithm for fixed allocations. \\
\hline
~\citet{brustle2020one} 
& Indivisible goods 
& Envy-freeness
& Envy-free allocations with bounded per-agent subsidies. \\
\hline
~\citet{goko2024fair} 
& Indivisible goods
&Envy-freeness, truthful
& Truthful EF mechanism with bounded subsidies. \\
\hline
~\citet{kawase2024towards} 
&Indivisible goods
& Envy-freeness
& Polytime conversion from EF1 to EF with improved bounds. \\
\hline
~\citet{wu2023one}, ~\citet{wu2024tree}
& Goods \& chores
& Proportionality 
& Tight and improved subsidy bounds for proportionality. \\
\hline
~\citet{dai2024weighted} 
& House allocation
& Weighted\ envy-freeness 
& Characterization and existence of weighted EF and EF-able assignments. \\
\hline
~\citet{liu2024mixed} 
&  Mixed fair division
& Various 
& Survey and open problems. \\
\bottomrule
\end{tabularx}
\caption{Comparison of representative related works on fair division with subsidies, highlighting problem settings, fairness notions, and main results.}
\label{tab:related_work}
\end{table}

In the specific context of house assignment—where each agent receives exactly one item—few works address envy-freeness. \citet{gan2019envy} proposed a polynomial-time algorithm to decide and compute envy-free assignments in such settings.

Our setting also draws on the literature addressing repeated allocations, where items are assigned multiple times rather than just once. Two main approaches have been considered: one, as proposed in our work, aims to preserve relative envy across repetitions; the other focuses on balancing envy over time. \citet{balan2011long} study long-term fairness by prioritizing the most disadvantaged agents in repeated decisions. \citet{igarashi2024repeated} introduce a model for repeated allocation of goods and chores, proving the existence of proportional and Pareto-optimal sequences under certain conditions, while noting that envy-freeness may require relaxed criteria or flexible repetition counts. To our knowledge, repeated fair allocation with monetary transfers has not yet been addressed.

Another relevant body of work concerns fair assignments with initial endowments and strategic considerations. \citet{yilmaz2010probabilistic} introduced the Probabilistic Serial mechanism to incorporate private endowments, ensuring ordinal efficiency and individual rationality via a generalized eating algorithm \citep{bogomolnaia2001new}, though at the cost of strategy-proofness. More recently, \citet{thomson2024allocation, thomson2024manipulability} demonstrated that most standard mechanisms are manipulable, even under assumptions such as homothetic preferences or penalties for dishonesty. In the housing context, \citet{abdulkadirouglu1999house} designed a mechanism that is strategy-proof, Pareto-efficient, and respects existing tenants’ rights. Finally, \citet{segal2022redividing} addressed re-division problems, proposing mechanisms that balance fairness with respect for ownership, and analyzing trade-offs using the price-of-fairness metric.

See \Cref{app:related-work} for further related work.

\renewcommand\cellalign{l}

\begin{toappendix}
\subsection{Further Related Work} \label{app:related-work}

\paragraph{House assignment.}
The house assignment problem, also referred to as the assignment problem, involves allocating a set of $m$ houses to a set of $n$ agents based on their preferences, ensuring each agent receives exactly one house \cite{hylland1979efficient, zhou1990conjecture}. 
The problem is often studied under the lens of economic efficiency, where Pareto optimality is a common goal—ensuring that no reassignment improves some agents’ outcomes without worsening others, as noted by Abraham et al. \cite{abraham2004pareto} in their study on Pareto optimality in house assignment problems, and Manlove in \cite{manlove2013algorithmics}.

Initial studies predominantly focused on achieving strategy-proofness and stability in house assignment mechanisms. For instance, Svensson \cite{svensson1999strategy} provided characterizations for strategy-proof assignments under strict preferences. Building on this foundation, Abdulkadiroğlu and Sönmez \cite{abdulkadirouglu1999house} explored mechanisms such as deferred acceptance and top trading cycles, emphasizing the trade-offs between fairness and efficiency.

Recent work has shifted towards fairness considerations. Gan et al. \cite{gan2019envy} present a polynomial-time algorithm that determines whether an envy-free assignment exists, and if so, computes one such assignment. Additionally, they show that an envy-free assignment exists with high probability if the number of houses exceeds the number of agents by a logarithmic factor. 
Kamiyama et al. \cite{kamiyama2021complexity} study fairness in house assignment, showing that maximizing the number of envy-free agents is hard to approximate within a factor of 
$n^{1-\alpha}$ for any constant $\alpha > 0$. They prove that the exact problem is NP-hard, even for binary utilities, while deciding the existence of proportional assignments is computationally hard, and equitability can be efficiently determined.

In the context of binary preferences, where agents classify houses as acceptable or unacceptable, Aigner-Horev and Segal-Halevi \cite{aigner2022envy} addressed the maximum envy-free partial matching problem, proposing efficient algorithms for special cases. In the context of dichotomous preferences, where agents classify houses as acceptable or unacceptable, Aziz \cite{aziz2016mechanisms} proposed strategy-proof and polynomial-time algorithms to maximize the number of agents receiving acceptable houses while ensuring individual rationality. Additionally, his work addressed scenarios with existing tenants, achieving core stability in these settings.

 Hosseini et al. \cite{hosseini2023graphical} extend the classical house assignment problem by introducing a graphical setting where agents are positioned at the vertices of a social network. In this model, an agent only experiences envy toward its neighbors. Their work focuses on minimizing the total pairwise envy across all edges in the social graph. They provide structural insights and computational results for various graph classes, including disjoint unions of paths, cycles, and more.

\paragraph{Truthfulness Relaxations.}
Several studies have explored relaxations of truthfulness that restrict attention to specific types of manipulations deemed more plausible. These relaxations ensure that no agent benefits from any manipulation within the designated subset. The precise definition of "likely" manipulations varies across different approaches.

Brams et al. \cite{brams2006better} introduced maximin strategy-proofness, a concept in which an agent chooses to manipulate only when it is guaranteed to be beneficial in every scenario. Waxman et al. \cite{waxman2021manipulation} were, to the best of our knowledge, the first to explicitly use the term "safe manipulation". Their work analyzed agents' manipulation strategies under three different levels of knowledge in social networks.

In the context of cake-cutting, Bu et al. \cite{buOnExistence} defined a mechanism as risk-averse truthful (RAT) if no agent can manipulate in a way that is both safe and profitable. 
Recently, Hartman et al. \cite{hartman2025s} extended RAT beyond cake-cutting to general social choice settings and proposed a quantitative measure to evaluate a mechanism’s resilience to manipulation. They introduced the RAT-degree, which quantifies the minimum number of agents whose reports must be known to enable safe manipulation, thereby bridging the gap between classic truthfulness and RAT.

Another related concept, not-obvious manipulability (NOM), was introduced by Troyan and Morrill \cite{troyan2020obvious}. This framework assumes that agents consider only extreme outcomes—the best or worst cases—when deciding whether to manipulate. RAT and NOM are independent notions. 
\end{toappendix}

\section{Preliminaries}
\label{section:model}

\paragraph{Agents and valuations.}
We denote by $[t]$ the set $\{1, 2, ..., t\}$ for any positive integer $t$.
There are $n$ agents; the set of agents is denoted $N = [n]$.
The set of old apartments is denoted $O = \{o_1,\ldots,o_n\}$, where $o_i$ is the old apartment owned by agent $i$.
The set of new apartments is denoted $M = \{a_1,...,a_n\}$.

An \textit{assignment} is a bijective mapping from agents to new apartments. It is denoted $A = (A_1,\ldots,A_n)$, where $\forall {i\in N}: A_{i} \in M$, and $A_i \neq A_j$ for all $i \neq j \in N$, and $A_i$ is the new apartment assigned to agent $i$.

Each agent $i \in N$ has a \textit{valuation function} $v_{i} : M \cup O \rightarrow \mathbb{R}_{\geq 0}$, indicating how much they value different apartments. 
So $v_i(A_i)$ represents agent $i$'s valuation of their new apartment under assignment $A$. 
Note that the classic model is equivalent to assuming that all agents value all old apartments at $0$.

In addition to receiving item $A_i$, each agent $i \in N$ is given a payment $p_i$,
where the vector of payments is denoted as $\price{} = (p_1 ,..., p_n).$ These payments can be either positive or negative. A payment vector $\price{}$ is said to be \emph{balanced} if $\sum_{i\in N} p_i = 0$.

We assume that agents are \emph{quasilinear}, so that the utility of each agent for an apartment and payment is
$$
u_i(A_i,p_i) := 
v_{i}(A_{i}) + p_{i}
.$$
Moreover, we assume that the subjective happiness of each agent is determined by how much his new bundle (new apartment plus payment) is better than his old apartment. We measure this improvement by the \emph{difference}:$^{\ref{footnote:ratio-model}}$
\begin{align*}
d_i(a, o, p) := u_i(a,p) - v_i(o)
=
v_{i}(a) - v_i(o) + p.
\end{align*}
So the subjective happiness of agent $i$ given allocation $A$ and payment vector $\price$ is $d_i(A_i,o_i,p_i)$; and the subjective happiness that $i$ attributes to some other agent $j$ is $d_i(A_j, o_j, p_j)$.
We also denote $d_i(a,o) := d_i(a,o,0)$.

\paragraph{Utilitarian social welfare.}
The \textit{utilitarian social welfare} of an assignment $A$ is $\sum_{i\in N} v_{i}(A_{i})$. An assignment $A$ of the new apartments is called \emph{utilitarian welfare maximizing} if $\sum_{i\in N}v_{i}(A_{i})\geq \sum_{i\in N}v_{i}(B_{i})$ for any other assignment $B$.

\section{Minimizing Envy}
\label{section:max-envy}
Given an assignment $A$ and payment-vector $\price{}$, we define the \emph{envy} felt by an agent $i$ by:
\begin{align*}
ENVY_i(A,\price{}) := \max_{j\in N} d_i(A_j,o_j,p_j) - d_i(A_i,o_i,p_i).
\end{align*}
The pair $(A,\price{})$ is called \emph{Envy Free (EF)} if $ENVY_i(A,\price{})\leq 0$ for all $i\in N$. Equivalently, $d_i(A_i,o_i,p_i)\geq d_i(A_j,o_j,p_j)$ for all $i,j\in N$.
An assignment $A$ is called \emph{EF-able} if there exists a payment vector $\price{}$ such that $(A,\price{})$ is EF.~
Note that this definition reduces to the standard definition of envy-freeness when there are no old apartments (equivalently, when $v_i(o_j)=0$ for all $i,j\in N$).

\subsection{The Envy Graph}
To analyze the conditions for EF-ability, we extend the notion of \emph{envy-graph}, 
which was introduced by \citet{aragones1995derivation}
and studied also by \citet{halpern2019fair}.

The \emph{envy-graph} of an assignment $A$ and a payment vector $\mathbf{p}$, denoted $G_{A,\mathbf{p}}$, is a complete directed graph in which the set of nodes is the set $N$ of agents. Each pair of agents is connected by arcs in both directions.
For any pair of agents $i,j\in N$, the arc $(i,j)$ has a \emph{cost} defined as the envy $i$ feels toward $j$ under assignment $A$, that is, $cost_A(i,j) := \left(v_i(A_j) + p_j \right) - \left(v_i(A_i) + p_i\right)$.
We denote the cost of a path $(i_1,...,i_k)$ as $cost_A(i_1,...,i_k) = \sum_{j = 1}^{k-1} cost_A(i_j, i_{j+1})$. 

For convenience, when the payment vector is zero (i.e., no payments are made), we denote it by $G_A$.

In addition to the envy-graph $G_A$ of the new apartments assignment, we are also interested in the following envy-graphs:
\begin{itemize}
\item \textbf{The original apartments graph, $G_{O}$}, where the cost of every arc $(i,j)$ is $cost_O(i,j) := v_i(o_j) - v_i(o_i)$;

\item \textbf{The difference graph, $G_{A,O}$}, where 
the cost of every arc $(i,j)$ is $(v_i(A_j) - v_i(o_j)) - ( v_{i}(A_i) - v_i(o_i)) = cost_A(i,j) - cost_O(i,j)$.
If the weights on each graph are represented by a matrix, then the matrix of $G_{A,O}$ is simply the difference between the matrix of $G_A$ and the matrix of $G_O$; 
shortly, $G_{A,O} = G_{A} - G_{O}$.
\end{itemize}

Using the difference envy-graph, we can extend the characterization of \citet{halpern2019fair} to the Reconstruct and Divide setting.
\begin{theorem} 
\label{theorem:mefable-iff-cost-cycles}
For $O$ the assignment of the original apartments and $A$ the new apartments assignment, the following are equivalent:
\begin{enumerate} [label=(\alph*)]
\item $A$ is EF-able;
\item $
\sum_{i\in N} v_i(A_i) - v_i(o_i) \geq \sum_{i\in N} v_i(A_{\pi(i)}) - v_i(o_{\pi(i)})$, 
for all permutations $\pi$ of $N$;
\item The difference graph $G_{A,O}$ has no positive-cost cycles; equivalently, for every directed cycle $C$, $cost_A(C) \leq cost_O(C)$. \label{condition b of theorem theorem:mefable-iff-cost-cycles}
\end{enumerate}
\end{theorem} 
\begin{proof}
$(a) \Rightarrow (b)$: Suppose $A$ is EF-able. Then, there exists a payment vector \textbf{p} such that$(A,\price{})$ is EF. that is, for all agents $i,j$, $v_i(A_i) - v_i(o_i) + p_i \geq v_i(A_j) - v_i(o_j) + p_j$. Equivalently, $$\left( v_i(A_j) - v_i(o_j) \right) - \left( v_i(A_i) - v_i(o_i) \right) \leq p_i - p_j.$$

Consider any permutation $\pi$ of $N$. Then,
\begin{align*}
    & \sum_{i\in N} \left(\left( v_i(A_{\pi(i)}) - v_i(o_{\pi(i)}) \right) - \left( v_i(A_i) - v_i(o_i) \right) \right) \leq \\
    &\sum_{i\in N} \left( p_i - p_{\pi(i)}\right) = 0.
\end{align*}
The last entry is zero as all the payments are considered twice, and they cancel out each other. Hence, 
\[
\sum_{i\in N} v_i(A_{i}) - v_i(o_{i}) \geq \sum_{i\in N} v_i(A_{\pi(i)}) - v_i(o_{\pi(i)})
\]
for each permutation $\pi$ of $N$.

$(b) \Rightarrow (c)$: Suppose that $
\sum_{i\in N} v_i(A_i) - v_i(o_i) \geq \sum_{i\in N} v_i(A_{\pi(i)}) - v_i(o_{\pi(i)})$, 
for some assignment $A$, and all permutations $\pi$ of $N$. 
Consider a cycle $C = (i_{1} , ... , i_{r})$ in $G_{A,O}$, and a permutation $\pi$, defined for each agent $i_k \in N$ as follows: \[\pi(i_k) = \begin{cases}
    i_k, &i_k\notin C \\
    i_{k+1}, & k \in \{1, \ldots, r-1 \} \\
    i_1 & k = r 
\end{cases}. \]
Examining the cost of $C$ in the difference graph $G_{A,O}$:
\begin{equation*} 
\begin{split}
    &cost_A(C) - cost_O(C) = cost_{A,O}(C) = \\
    &\sum_{i\in C} \left(\left(v_{i}(A_{\pi(i)}) - v_i(o_{\pi(i)}) \right) - \left(v_{i}(A_{i}) - v_{i}(o_i) \right)\right) \leq 0.
\end{split}
\end{equation*}
The validity of the equality arises from the assumption.

$(c) \Rightarrow (a)$: Suppose $G_{A,O}$ has no positive-cost cycles. Then, we can define for each agent $i$, the maximum-cost of any path in the difference graph that starts at $i$. We denote the cost of this path by $\ell_i$. For each $i \in N $, let $p_i = \ell_i - \frac{\sum_{i\in N} \ell_i}{n}$. 
It is noteworthy that $\sum_{i\in N}p_{i} = 0$, establishing the suitability of \textbf{p} as a balanced payment vector.
Furthermore, in accordance with the definition of highest-cost paths, it follows that for all $i \neq j \in N$:
\begin{align*}
     &p_i =  \ell_i - \frac{\sum_{i\in N} \ell_i}{n}  \geq 
     cost_{A,O}(i,j) + \ell_j -\frac{\sum_{i\in N} \ell_i}{n} = \\ &
     \left( v_i(A_j) - v_i(o_j) \right) - \left( v_i(A_i) - v_i(o_i) \right) + \ell_j - \frac{\sum_{i\in N} \ell_i}{n} =  \\ &
     v_i(A_j) - v_i(o_j) + p_j - \left(v_i(A_i) - v_i(o_i) \right) \Longleftrightarrow \\ & v_i(A_i) - v_i(o_i) + p_i \geq v_i(A_j) - v_i(o_j) + p_j.
\end{align*}

Hence, $(A,\price{})$ are envy-free, and thus, $A$ is EF-able.
\end{proof}

With no old apartments, any utilitarian-welfare-maximizing assignment welfare is EF-able \citep{haake2002bidding,sung2004competitive}.
We extend this as follows.
\begin{proposition} \label{prop:sufficient-conditions-for-mef}
Let $A$ be a utilitarian-welfare-maximizing assignment. Then $A$ is EF-able if one of the following holds:

(a) The original envy-graph $G_O$ has no negative-cost cycles;

(b) All $n$ agents assign the same values (e.g. market prices) to the original apartments.
\end{proposition}

\begin{proof}
(a) If $A$ is utilitarian-welfare-maximizing, then $G_A$ has no positive-cost cycles \citep{halpern2019fair}.
For each directed cycle $C$, its cost in $G_{A,O}$ equals its cost in $G_A$ minus its cost in $G_O$. As $G_O$ has no negative-cost cycles, the cost of $C$ in $G_{A,O}$ is non-positive. Hence $A$ is EF-able by \Cref{theorem:mefable-iff-cost-cycles}. 

(b) Under the given conditions, the cost of every directed cycle in $O$ is zero. Hence $A$ is EF-able by part (a).
\end{proof}

On the negative side, there are cases in which no assignment of the new apartments is EF-able.
The following example formalizes the intuitive explanation given in Example \ref{example:no-ef}.
\begin{Example} 
\label[example]{example:no-efable}
Suppose there are two agents $i_1$ and $i_2$ with valuations:
\[
\begin{bmatrix}
   & o_1 & o_2 & & a_1 & a_2 \\
  i_1 & z + \epsilon & z & & V_1 & V_1 \\
  i_2 & y & y + \delta & & V_2 & V_2 \\
\end{bmatrix}.
\]
when $y,z,\epsilon,V_1, V_2 > 0$.
Let $C$ be the cycle $i_1 \to  i_2 \to i_1$.
Then $cost_O(C) = -(\epsilon + \delta)$.
In contrast, for any assignment $A$ of the new apartments, $cost_A(C) = 0$.
Hence, $$cost_{A,O}(C) = cost_A(C) - cost_O(C) = \epsilon + \delta > 0.$$ By \Cref{theorem:mefable-iff-cost-cycles}, $A$ is not EF-able. \qed
\end{Example}

\subsection{Max average cycle cost and min envy} \label{sec: max mean cycle cost}
Since EF-able assignments may not exist (Example \ref{example:no-efable}), our goal shifts to minimizing the largest envy experienced by an agent. 
We begin with a lemma demonstrating that envy may remain despite efforts to eliminate it using balanced payments.
\begin{lemma}
	 \label[lemma]{lemm: cycle cost unchanged}
    The total cost of any cycle in $G_{A,O}$ is the same for any payment vector.
\end{lemma}
\begin{proof}
    Suppose we give some agent $i$ a payment of $p_i$. As a result, the cost of every edge from $i$ decreases by $p_i$ (as $i$ experiences less envy), and the cost of every edge into $i$ increases by $p_i$ (as other agents experience more envy in $i$).
Every cycle through $i$ contains exactly one edge from $i$ and one edge into $i$, and every other cycle contains no such edges. Therefore, the total cost of any cycle remains unchanged.
\end{proof}

Now, we demonstrate how to construct a payment vector for a given assignment that minimizes the largest envy, that is,
we solve the problem $\displaystyle \min_{\price{}} \max_{i\in N} ENVY_i(A,\price{})$. 
~
Although this can be done by solving a linear program,
our algorithm runs in strongly-polynomial time and also reveals interesting links between our problem and fundamental graph-theoretic concepts.
\begin{Definition}
Given a directed graph $G$ with edge costs, 
The \emph{average cost} of a path is the total cost of the path divided by the number of edges in it.
The \emph{maximum average cycle cost (MACC)} of $G$ is the maximum, over all directed cycles $C$ in $G$, of the average cost of $C$.
\end{Definition}
In the literature, the MACC of $G$ is  also known as its \emph{maximum mean cycle weight} \cite{karp1978characterization}.
\no{The MACC of any given directed graph can be computed in strongly-polynomial time (\cite{karp1978characterization,v1982optimal}).}
The following lemma relates the MACC to the minimum attainable envy.

\begin{lemma} \label{lemm: envy bound mean cycle}
For a given assignment $A$, it is possible to compute in polynomial time a payment vector $\price{}$ such that the envy between any two agents is bounded by the maximum average cycle cost of $G_{A,O}$. Specifically, for all $i,j \in N$, the following holds:
$$
d_i(A_j,o_j,p_j) -d_i(A_i,o_i,p_i)
\leq c,$$ 
where $c := MACC(G_{A,O})$.
\end{lemma}
\begin{proof} 
Let $C$ be a cycle with the maximum average cost in $G_{A,O}$.
meaning that for any other cycle $C'$, it holds that $c = \frac{cost_A(C)}{|C|} \geq \frac{cost_A(C')}{|C'|}$.

Temporarily modify $G_{A,O}$ to obtain $G'_{A,O}$, where the cost of each edge is its original cost in $G_{A,O}$ minus $c$. The cost of every cycle $C'$ in $G'_{A,O}$ is thus its cost in $G_{A,O}$ minus $c   |C'|$, which is greater than $\frac{cost_A(C')}{|C'|}   |C'| = cost_A(C')$, meaning there are no positive-cost cycles:
\begin{align*}
    &cost_{A}(C') - c|C'| \leq cost_{A}(C') - \frac{cost_A(C')}{|C'|} |C'| = \\
    &cost_A(C') - cost_A(C') = 0.
\end{align*}

\no{By \Cref{theorem:mefable-iff-cost-cycles}, there exists a payment vector $\price{}$ that eliminates all envy in $G'_{A,O}$. Specifically, for each agent $i \in N$, define $$p_i = \ell_i - \frac{\sum_{i\in N} \ell_i}{n},$$ where $\ell_i$ denotes the cost of the maximum-cost path in $G'_{A,O}$ originating from $i$. Moreover, such a vector can be computed in polynomial time: initially, the Floyd-Marshall algorithm 
(\citet{weisstein2008floyd, wimmer2017floyd}) is applied to the graph derived by negating all edge costs in $G'_{A,O}$ (This has a linear time solution since there are no cycles with positive costs in the graph). Hence, determining the longest path cost between any two agents, accomplished in $O(nm+n^{3})$ time. Subsequently, the longest path starting at each agent is identified in $O(n^{2})$ time.}
Using this payment vector, we have for each $i,j \in N$:
\begin{align*}
    & d_i(A_j,o_j,p_j) \leq d_i(A_i,o_i,p_i) + c \Longleftrightarrow d_i(A_j,o_j,p_j)  - d_i(A_i,o_i,p_i) \leq c.
\end{align*}
Hence, $ENVY_i(A,\price{}) \leq c$ for all $i\in N$.
\end{proof}

\begin{lemma} 
	\label[lemma]{lemm: smallest maximum envy}
Let $A$ be an assignment of the new apartments.
Then 
\begin{align*}
 \min_{\price{}} \max_{i\in N} ENVY_i(A,\price{}) = MACC(G_{A,O}),
\end{align*}
and a minimizing $\price$ can be computed in polynomial time.
\end{lemma}
\begin{proof}
Let $C$ be a maximum average-cost cycle in $G_{A,O}$, and let $c$ be its average cost.
~
By \Cref{lemm: cycle cost unchanged}, 
the cost of a cycle remains the same under any payment vector $\price{}$, and the same holds for the average cost. 
By the pigeonhole principle, the cycle has at least one edge with cost at least $c$; this means that, for any payment vector, the maximum envy in $A$ is at least $c$.
This proves the $\geq$ direction.

By \Cref{lemm: envy bound mean cycle}, a payment vector for which the envy of all agents is at most $c$ exists and can be computed in polynomial time; this proves the $\leq$ direction.
\end{proof}

To illustrate \Cref{lemm: smallest maximum envy}, consider \Cref{example:no-efable} again. 
There is only one directed cycle, namely $i_1\to i_2 \to i_1$; its cost in $G_{A,O}$ is $\epsilon + \delta$, so its average cost is $(\epsilon + \delta)/2$. Indeed, with the payment vector $p_1=(-\epsilon + \delta)/2,p_2=(\epsilon - \delta)/2$, the envy of agent 1 modifies from $\epsilon$ to $(\epsilon+\delta)/2$ and the envy of agent 2 modifies from $\delta$ to $(\epsilon+\delta)/2$, so both agents attain the envy bound. With every other payment vector, one agent would experience envy higher than $(\epsilon+\delta)/2$.

\no{As noted earlier, the MACC of any directed graph can be computed in strongly polynomial time.}
Hence, by \Cref{lemm: smallest maximum envy}, we can find $\price{}$ minimizing $\max_{i\in N} ENVY_i(A,\price{})$ in strongly-polynomial time.

However, minimizing the largest envy over all assignments is much more challenging, as by \Cref{lemm: smallest maximum envy}, it requires to find an assignment $A$ that minimizes $MACC(G_{A,O})$.
\begin{Open} \label{open question: macc}
Is there a polynomial-time algorithm that, given an assignment $O$ of the original apartments, computes an assignment $A$ of the new apartments that minimizes $MACC(G_{A,O})$?
\end{Open}

%
\section{Minimizing Disproportionality}
\label{section:proportionality}
Given an assignment $A$ and a payment vector $\price{}$, 
we define the \emph{total improvement} for agent $i$ as
\begin{align*}
	TOTAL_i := \sum_{j\in N}d_i(A_j,o_j,p_j).
\end{align*}
As $p$ is a balanced payment vector the payments cancel out, so an equivalent definition is $TOTAL_i = \sum_j (v_i(A_j) - v_i(o_j))$, which is simply the sum of all new apartments' values minus the sum of all old apartments' values in $i$'s eyes; hence, $TOTAL_i$ does not depend on $\price$ nor on $A$.

An allocation $(A,\price)$ is called \emph{proportional} if 
each agent enjoys at least a fraction $1/n$ of the total improvement, that is,  $d_i(A_i,o_i,p_i)\geq \frac{1}{n}TOTAL_i$ for all $i
\in N$.

An assignment $A$ is called PROP-able if there exists a payment vector $\price{}$ such that $(A,\price{})$ is proportional.

It is well-known in other fair division domains that envy-freeness implies proportionality, and when $n=2$ the opposite implication holds too. The same implications exist in our domain.
\begin{proposition}
(a) If an allocation $A$ and a payment vector $\mathbf{p}$ are envy-free, then they are also proportional.

(b) When $n=2$, if $(A,\mathbf{p})$ is proportional, then it is also envy-free.
\end{proposition}
\begin{proof}
(a) If $(A,\price)$ is envy-free, then by definition, for all $i\in N$,
 $$d_i(A_i,o_i,p_i) = \max_{j\in N} d_i(A_j,o_j,p_j).$$
As the maximum is always at least as high as the average, this implies $d_i(A_i,o_i,p_i)\geq TOTAL_i/n$, so $(A,\price)$ is proportional.

(b) When there are only two agents, if
$d_i(A_i,o_i,p_i)$ is at least as large as the average, then it must be maximum, so $(A,\price)$ is EF. 
\end{proof}

A PROP-able allocation may still not exist.
This is shown by the same \Cref{example:no-efable}, as in this example there are two agents.
Hence, as in the previous section, we aim to minimize the largest \emph{deviation} from proportionality, which we define by $DP_i(A,\price) := \frac{1}{n}TOTAL_i - d_i(A_i,o_i,p_i)$. 

We start by minimizing the largest disproportionality for a given assignment, that is, solving $\displaystyle \min_{\price{}} \max_{i\in N} DP_i(A,\price{})$.

We denote by $DP_i(A)$ the disproportionality of $i$ when the assignment is $A$ and all payments are $0$. 
Note that $DP_i(A) = \frac1n [\sum_{j\neq i} d_i(A_j,o_j) - (n-1) d_i(A_i,o_i)]$.

We denote the \emph{total disproportionality} of assignment $A$ by $\displaystyle DP_N(A) := \sum_{i\in N} DP_i(A)$.
\begin{lemma} \label{lem:smallest-disproportionality}
Let $A$ be an assignment of the new apartments.
Then 
\begin{align*}
 \min_{\price{}} \max_{i\in N} DP_i(A,\price{}) = \frac{DP_N(A)}{n},
\end{align*}
and the minimizing $\price$ is given by 
\begin{align} 
	\label{eq:dp-payments}
    p_i &= DP_i(A) - \frac{DP_N(A)}{n} 
\end{align}
\end{lemma}

\begin{proof}
When the payment to agent $i$ increases by $1$,
$d_i$ increases by $1$ whereas $TOTAL_i/n$ increases by $1/n$; hence the disproportionality of $i$ decreases by $(n-1)/n$ and the disproportionality of any other agent $j$ increases by $1/n$. Hence:
\begin{align}
\notag
	DP_i(A,\price)
	&=
	DP_i(A) - \frac{n-1}{n}p_i + \frac{1}{n}\sum_{j\neq i} p_j
\\
\label{eq:dpi}
	&=
	DP_i(A) - p_i + \frac{1}{n}\sum_{j\in N} p_j.
\end{align}
Hence, the sum $\displaystyle \sum_{i\in N} DP_i(A,\price) = DP_N(A)$ for any $\price$ (balanced or not).
By the pigeonhole principle, there must be an agent $i$ for whom $DP_i(A,\price) \geq \frac{DP_N(A)}{n}$. 
This proves the $\geq$ direction.

Now, let us compute the disproportionality of agent $i$ that when $\price$ is given by \eqref{eq:dp-payments}. As this $\price$ is balanced, \eqref{eq:dpi} gives 
\begin{align*}
	DP_i(A,\price{}) 
	& =
	DP_i(A) - p_i
	\\
	& =
	DP_i(A) - DP_i(A) + \frac{DP_N(A)}{n} = \frac{DP_N(A)}{n}.
\end{align*}
This proves the $\leq$ direction.
\end{proof}

\begin{toappendix}
\subsection{Explicit formula for the payment}
\label{eq:dp-payments-explicit}

The payment given to each agent $i$ by the Minimum Disproportionality mechanism is:
\begin{align*} 
	p_i &= DP_i(A) - \frac{DP_N(A)}{n} 
	\\
	&=
	\frac{1}{n^2}\bigg[
	(n-1)\sum_{j\neq i}d_i(A_j,o_j)
	- (n-1)^2 d_i(A_i,o_i) \\
	&- \sum_{j\neq i}\sum_{k\neq  j} d_j(A_k,o_k) + \sum_{j\neq i}(n-1)d_j(A_j,o_j)
	\bigg]    
	\\
	&=
	\frac{1}{n^2}\bigg[
	(n-1)\sum_{j\neq i}v_i(A_j)
	- (n-1)^2 v_i(A_i) \\
	&- \sum_{j\neq i}\sum_{k\neq  j} v_j(A_k) + \sum_{j\neq i}(n-1)v_j(A_j)
	\bigg]
	\\
&-
\frac{1}{n^2}\bigg[
(n-1)\sum_{j\neq i}v_i(o_j)
- (n-1)^2 v_i(o_i) \\
& - \sum_{j\neq i}\sum_{k\neq  j} v_j(o_k) + \sum_{j\neq i}(n-1)v_j(o_j)
\bigg].
\end{align*}
This can be written as    
\begin{align*}
p_i = \frac{1}{n^2} \Bigg[p^A_i - p^O_i\Bigg],
\end{align*}
where $p^A$ and $p^O$ represent the contributions of new and old apartments respectively:
\begin{align*}
p^A_i &= (n-1)\sum_{j\neq i}v_i(A_j)
- (n-1)^2 v_i(A_i)
- \sum_{j\neq i}\sum_{k\neq  j} v_j(A_k) + \sum_{j\neq i}(n-1)v_j(A_j)
\\
p^O_i &= (n-1)\sum_{j\neq i}v_i(o_j)
- (n-1)^2 v_i(o_i)
- \sum_{j\neq i}\sum_{k\neq  j} v_j(o_k) + \sum_{j\neq i}(n-1)v_j(o_j)
\end{align*}
\end{toappendix}

Note that the payment vector given by \eqref{eq:dp-payments} is balanced.

\Cref{lem:smallest-disproportionality} implies that,
in order to minimize the largest disproportionality over all assignments, one should find an assignment $A$ that minimizes $DP_N(A)$. 
We show below that this minimum is attained by any utilitarian-welfare-maximizing assignment.
\begin{lemma}
\label{lemma:utilitarian-minimizes-envy-sum}
If an assignment $A$ maximizes the utilitarian social welfare, then it minimizes $DP_N(A)$ over all assignments.

Hence, $A$ minimizes $\min_{\price{}} \max_{i\in N} DP_i(A,\price{})$.
\end{lemma}
\begin{proof}
    Let $A$ and $B$ be assignments in which each agent receives exactly one item. Since $A$ maximizes the utilitarian social welfare, we have $\sum_{i\in N} v_i(A_i) \geq \sum_{i\in N} v_i(B_i)$.
~
Now, consider the difference in total disproportionality between $A$ and $B$:
\begin{align*}
    & DP_N(A) - DP_N(B) = 
    \sum_{i\in N} (DP_i(A) - DP_i(B)) = \\
    & \frac{1}{n} \sum_{i\in N} \left( \sum_{j\in N} v_i(A_j) - v_i(o_j) - v_i(A_i) + v_i(o_i) \right) - \\
    &
    \frac{1}{n} \sum_{i\in N} \left( \sum_{j\in N} v_i(B_j) - v_i(o_j) - v_i(B_i) + v_i(o_i) \right) = \\
    & 
    \frac{1}{n} \sum_{i\in N} \left( v_i(M) -  v_i(O) -n v_i(A_i) + n v_i(o_i) \right) - \\
    &
    \frac{1}{n} \sum_{i\in N} \left( v_i(M) - v_i(O) - n v_i(B_i) + n v_i(o_i) \right) = \\
    & 
    \sum_{i\in N} v_i(B_i) - \sum_{i\in N} v_i(A_i) \leq 0.
\end{align*}
Hence, $DP_N(A)\leq DP_N(B)$, which proves the lemma.
\end{proof}

Consequently, we propose a polynomial-time mechanism for Reconstruct and Divide projects that computes an assignment and payment vector minimizing the maximum disproportionality among all agents:
\begin{algorithm}
\caption{Minimum Disproportionality Mechanism
\label{def:min-envy-sum}
}
\begin{algorithmic}[1]
\State Find an assignment $A$ maximizing the utilitarian welfare;
\State Compute a payment vector $\price{}(A)$ by Equation \eqref{eq:dp-payments}.
\end{algorithmic}
\end{algorithm}

As a corollary, we get that the existence of proportional allocations can be decided in polynomial time.
\begin{Corollary} \label{min ES mechanism}
For $O$, the original apartments, and $M$, the new apartments, an assignment $A$ and a balanced payment vector $\price{}(A)$ such that $(A,\price{}(A))$ is proportional can be found, or it can be determined that no such assignment and vector exist, all in polynomial time.
\end{Corollary}
\begin{proof}
Use \Cref{def:min-envy-sum} to compute $(A,\price)$.
Then compute $$\max_{i\in N} DP_i(A,\price{}).$$
~
If $\max_{i\in N} DP_i(A,\price{})\leq 0$, then $(A,\price)$ is proportional by definition.
~
Otherwise, no proportional allocation exists, as by \Cref{lemma:utilitarian-minimizes-envy-sum}, the attained disproportionality $\max_{i\in N} DP_i(A,\price{})$ is the smallest possible, so there is no other assignment $A'$ with $\max_{i\in N} DP_i(A',\price{})\leq 0$.
\end{proof}

\section{Strategic Manipulations} \label{sec: manipulations}

The Minimum Disproportionality Mechanism (\Cref{def:min-envy-sum}) takes the agents' valuations as input. Ideally, we would like the agents to report their true valuations. 

Formally, a \emph{manipulation} for a mechanism $\mathcal{M}$ by an agent $i \in N$ is an untruthful report $v_i'\neq v_i$. 
A manipulation is \emph{profitable} if there exists a set of reports $v_{-i}$ from the other agents such that the agent gains a higher utility than by misreporting:
\begin{align}
\label{eq:profitable}
\exists v_{-i} : v_i(\mathcal{M}(v_i', v_{-i})) > v_i(\mathcal{M}(v_i, v_{-i})) .
\end{align}
A mechanism $\mathcal{M}$ is \emph{truthful} if no agent has a profitable manipulation.

It is well-known that, even in the setting with no old apartments, no deterministic mechanism is truthful, budget-balanced and satisfies even weak fairness conditions
\citep{green1979coalition,zhou1990conjecture,dufton2011randomised}. In particular, if some agent $i$ wins some apartment $A_i$ when reporting truthfully, then a profitable manipulation for agent $i$ is to report a slightly lower value for $A_i$, as in some cases the assignment will not change but the payment for agent $i$ will increase.
However, such manipulations are usually \emph{unsafe}, as in some cases, reporting a lower value for $A_i$ might make the algorithm choose a different assignment, so agent $i$ would receive a worse apartment for a higher price.
Therefore, one can hope that  agents will not manipulate their valuations.

But with old apartments the situation is much worse: for each agent $i$, reporting a higher value for his old apartment $o_i$ is both profitable and \emph{safe}, as it has no effect on the assignment of new apartments, but it strictly increases $DP_i(A)$ and hence increases $p_i$ by \eqref{eq:dp-payments}. Therefore, even  agents will most probably manipulate their valuations.

Formally, we say that a manipulation is \emph{safe} if it never results in a worse outcome for the manipulator—i.e., the agent weakly prefers it over truthfulness for any possible reports of the other agents:
\begin{align}
\label{eq:safe}
\forall v_{-i} : v_i(\mathcal{M}(v_i', v_{-i})) \geq v_i(\mathcal{M}(v_i, v_{-i})).
\end{align}
A mechanism $\mathcal{M}$ is \emph{safely manipulable} if some agent has a manipulation that is both profitable and safe. Otherwise, $\mathcal{M}$ is \emph{Risk-Avoiding Truthful (RAT)} \citep{buOnExistence,hartman2025s}.

\subsection{Valuations based on Characteristics}
As \Cref{def:min-envy-sum} is not even  truthful,
we aim to improve its strategic properties by changing the method for eliciting agents' valuations. 
Instead of asking each agent to evaluate each apartment, we ask them to assess the value of shared \emph{characteristics} found in both old and new apartments, such as: directions of exposure, floor, parking, etc. This approach is inspired by the way actual real-estate appraisers compute the value of an estate.
Now, if an agent wants to increase the value of his old apartment, he has to increase the value of some of its characteristics; however, these characteristics might also be present in some new apartments, which might affect the assignment in a way that decreases the manipulator's utility.

Formally, we assume that there is a fixed set $T$ of potential apartment characteristics, such as floor level, parking availability, airflow direction, and natural light. Each agent assigns a score to each characteristic. 
For an apartment $a$ and a characteristic $t \in T$, define the indicator variable $\beta_{t,a}$, such that%
:
\[
\beta_{t,a} =
\begin{cases}
     & 1, \text{if apartment } a \text{ possesses characteristic } t \\
    & 0, \text{otherwise}
\end{cases}
\]

We consider two ways to aggregate the values of characteristics.

\paragraph{(1) Additive Characteristics.}
Here we assume that the value of an apartment is the sum of the values of its characteristics. 

For each agent $i\in N$ and each characteristic $t \in T$, let $\alpha_{i,t}$ represent the value assigned by agent $i$ to characteristic $t$. 
The valuation of apartment $a$ for agent $i$ is then given by
$v_{i}(a) = \sum_{t \in T}\alpha_{i,t}\cdot \beta_{t,a}.
$
See \Cref{exmample: sum of characteristics} for an example.

\paragraph{(2) Multiplicative Characteristics.}
Here we assume that each apartment has a base price determined by its size. Each agent specifies the percentage of the apartment's base price that they attribute to each characteristic.

For each agent $i\in N$ and each characteristic $t \in T$, let 
$\theta_{i,t}$ be the percentage assigned by agent $i$ to characteristic $t$, and let 
$\rho_{i,t} := 1 + \frac{\theta_{i,t}}{100}$.
Note that a characteristic might have a negative effect (e.g. some people do not like apartments in high floors). In that case we set $\theta_{i,t}$ to a negative amount between $-100$ and $0$, and get $0<\rho_{i,t}< 1$.
Equivalently, we could say $\theta$ is between 0 and 100 and define $\rho_{i,t}=1-\theta_{i,t}/100$.

For each apartment $a$, let $\phi_a$ denote its fixed base price, computed as the product of the apartment size in square meters, by the market price per square meter in the region.
The valuation of apartment $a$ for agent $i$ is then given by:
$
v_i(a) = 
\phi_a \prod_{t\in T}\rho_{i,t}^{\beta_{t,a}}.
$
The multiplicative method is closer to the one actually used by 
appraisers to determine the value of an apartment in Reconstruct and Divide projects.
~
See \Cref{example: Multiplicative Characteristics} for an example.

\subsection{Manipulations of Minimum Disproportionality mechanism}
Switching to characteristic-based evaluations does not guarantee that the Minimum Disproportionality mechanism is RAT. For instance, if agent $i$'s old apartment has a unique characteristic $t$, increasing $v_i(t)$ safely increase their payment without affecting the assignment (\Cref{safely manipulable: old apartment unique characteristics}). Similarly, if $t$ exists in both $o_i$ and all new apartments, $p_i$ still increases while the assignment remains unchanged (see \Cref{safely manipulable: same characteristics}).

However, it is possible that the characteristics of the new apartment are not all known to the agents, as the valuations are elicited from the agents before the new apartments are even built. To handle this issue, we extend the definition of a safe manipulation \eqref{eq:safe} to require that the manipulation is not harmful for the agent for any combination of characteristics of the new apartments. 

We show that, in this case, the Minimum Disproportionality mechanism has no safe manipulations.

For brevity, throughout this section we denote the true valuation vector $(v_1,v_2,\ldots, v_n)$ as $v$, and the manipulated valuation vector, where only agent 1 misreports, as $v' = (v_1',v_2,\ldots, v_n)$.
We use $\price{}(v)$ to represent the payment vector under the valuation profile $v$, and let $DP(A,v)$ denote the total disproportionalitycorresponding to assignment $A$ and valuation vector $v$.
Let $T(a)$ denote the set of characteristics of apartment $a$.

\begin{toappendix}
\subsection{Example: Additive Characteristics}
\label{exmample: sum of characteristics}
\begin{Example} 
    Consider 2 agent and 4 apartment's characteristics: direction of airflow, natural light, parking availability and floor level higher than 5.
    The values assigned by the agents to these characteristics are as follows:
    \[
    \begin{bmatrix}
        & &\makecell{\textbf{Direction Of} \\ \textbf{Airflow}} & \makecell{\textbf{Natural} \\ \textbf{Light}} & \makecell{\textbf{Parking}\\ \textbf{Availability}} & \makecell{\textbf{Floor Level} \\ \textbf{$>$ 5}} \\
        &\textbf{Agent 1} & 100 & 50 & 140 & 60 \\
        &\textbf{Agent 2} & 75 & 100 & 150 & 100
    \end{bmatrix}
    \]
    The indicator variables for the old apartments are defined as follows:
    \[
    \begin{bmatrix}
        & &\makecell{\textbf{Direction Of} \\ \textbf{Airflow}} & \makecell{\textbf{Natural} \\ \textbf{Light}} & \makecell{\textbf{Parking}\\ \textbf{Availability}} & \makecell{\textbf{Floor Level} \\ \textbf{$>$ 5}} \\
        &\textbf{$o_1$} & 1 & 0 & 0 & 1 \\
        &\textbf{$o_2$} & 1 & 1 & 0 & 1
    \end{bmatrix}
    \]
    

    
    The indicator variables for the new apartments are defined as follows:
    \[
    \begin{bmatrix}
        & &\makecell{\textbf{Direction Of} \\ \textbf{Airflow}} & \makecell{\textbf{Natural} \\ \textbf{Light}} & \makecell{\textbf{Parking}\\ \textbf{Availability}} & \makecell{\textbf{Floor Level} \\ \textbf{$>$ 5}} \\
        &\textbf{$a_1$} & 0 & 1 & 1 & 1 \\
        &\textbf{$a_2$} & 1 & 1 & 1 & 0
    \end{bmatrix}
    \]
    
    For example, apartment $a_1$ possesses the characteristics \textit{natural light}, \textit{parking availability} and \textit{floor level higher than 5}.
    The value assigned by agent 1 to the new apartment $a_1$ is the sum of the values agent 1 assigns to these characteristics is $$v_1(a_1) = 50 + 140 + 60 = 250.$$ 
%
 Similar calculations yield the following  values (in units) for each apartment:
    \[
    \begin{bmatrix}
        &                  & o_1 & o_2 & a_1 & a_2 \\
        & \textbf{Agent 1} & 160 & 210 & 250 & 290 \\
        & \textbf{Agent 2} & 175 & 275 & 350 & 325 \\
    \end{bmatrix}
    \]
This matrix can be used as an input to an assignment algorithm, such as the Minimum Disproportionality mechanism analyzed in \Cref{section:proportionality}.
\end{Example}

\subsection{Example: Multiplicative Characteristics} \label{example: Multiplicative Characteristics}
\begin{Example}
    We use the same characteristics as in \Cref{exmample: sum of characteristics}.
    The $\rho_{i,t}$ factors  assigned by the agents to the characteristics are as follows:
    \[
    \begin{bmatrix}
        & &\makecell{\textbf{Direction Of} \\ \textbf{Airflow}} & \makecell{\textbf{Natural} \\ \textbf{Light}} & \makecell{\textbf{Parking}\\ \textbf{Availability}} & \makecell{\textbf{Floor Level} \\ \textbf{$>$ 5}} \\
        &\textbf{Agent 1} & 1.05 & 1.03 & 1.07 & 1.06 \\
        &\textbf{Agent 2} & 1.04 & 1.05 & 1.08 & 1.05
    \end{bmatrix}
    \]

    The base prices of the apartments, are as follows:

\[
    \begin{bmatrix}
        & & \makecell{\textbf{base price}} \\
        & \textbf{$o_1$} & $160$ \\
        & \textbf{$o_2$} & $170$ \\
        & \textbf{$a_1$} & $210$ \\
        & \textbf{$a_2$} & $220$
    \end{bmatrix}
\]
    
    For example, apartment $a_1$ possesses the characteristics \textit{natural light}, \textit{parking availability} and \textit{floor level higher than 5}.
    The value assigned by agent 1 to the new apartment $a_1$ is $$v_1(a_1) = 210 \cdot 1.03 \cdot 1.07 \cdot 1.06 = 245.32$$ 
%
 %
%

    Similar calculations yield the following total values (in units) for each apartment:
    \[
    \begin{bmatrix}
        &                  & o_1 & o_2 & a_1 & a_2 \\
        & \textbf{Agent 1} & 178.08 & 194.86 & 245.32 &  254.58 \\
        & \textbf{Agent 2} & 174.42 & 194.91 & 250.04 & 259.45   \\
    \end{bmatrix}
    \]
    This matrix can be used as an input to an assignment algorithm, such as the Minimum Disproportionality mechanism analyzed in \Cref{section:proportionality}.
\end{Example}
\begin{proposition}
\label{safely manipulable: old apartment unique characteristics}
When old apartments have unique characteristics that new apartments lack, the Minimum Disproportionality mechanism
becomes safely manipulable. 

This is true both for additive and for multiplicative characteristics.
\end{proposition}
\begin{proof}
 We consider $n$ agents with arbitrary valuation functions and $T$ characteristics, where $t_1\in T$ applies only to $o_1$. Without loss of generality, the assignment that maximizes social welfare is $A_i = \{a_i\}$ for each $i\in N$.

 As shown above, the payment for each agent $i$ is given by:
     $$p_i = \frac{n-1}{n^2} DP_i(A) - \sum_{j\neq i \in N}\frac{ 1}{n^2}DP_j(A).$$
     Agent $1$ can benefit by misreporting their valuation, specifically by increasing the valuation of $t_1$. We denote the misreported valuation as $v'$. Since $t_1$ does not exist in the new apartments, this does not affect the assignment but does influence the payments in favor of agent 1:
     
     \begin{align*}
         & p(v')_1 = \frac{n-1}{n^2}DP_1(A,v') - \frac{1}{n^2}\sum_{j\neq 1\in N} DP_j(A,v') = 
         \\
         &
         \frac{n-1}{n^2}DP_1(A,v') - \frac{1}{n^2}\sum_{j\neq 1\in N} DP_j(A,v) =
         \\
         &
         \frac{n-1}{n^2}\left(DP_1(A,v) + n(v_1'(o_1) - v_1(o_1) ) \right) 
         - \frac{1}{n^2}\sum_{j\neq 1\in N} DP_j(A,v).
         \end{align*}
    Because $n^2> n-1\geq 1$ and $v_1(o_1) < v_1'(o_1)$ we receive that
         \begin{align*}
         p(v')_1 >
         \frac{n-1}{n^2}DP_1(A,v) - \frac{1}{n^2}\sum_{j\neq 1\in N} DP_j(A,v) = p(v)_1.
     \end{align*}
     Since the new assignment is not affected by the misreporting, no matter what the other agents report, agent 1 can benefit by misreporting their valuation, proving that the mechanism is safely manipulable.

     The proof for multiplicative characteristics is similar.
\end{proof}

\begin{proposition}
\label{safely manipulable: same characteristics}
Under additive characteristics,
when old apartments have a characteristic that appears in all new apartments, the Minimum Disproportionality mechanism becomes safely manipulable.
In contrast, this does not hold under multiplicative characteristics.
\end{proposition}

\begin{proof}
\underline{\bf Additive Characteristics:}
     We consider $n$ agents with arbitrary valuation functions and $T$ characteristics, where $t_1$ applies to $o_1, a_1, \ldots, a_n$ but not to $o_2,\ldots, o_n$. Without loss of generality, the assignment that maximizes social welfare is $A_i = \{a_i\}$ for each $i\in N$.

    As shown above, the payment for each agent $i$ is given by:
     $$p_i = \frac{n-1}{n^2}DP_i(A) - \frac{1}{n^2}\sum_{j\neq i\in N} DP_j(A).$$
     Agent $1$ can benefit by misreporting their valuation, specifically by increasing the valuation of $t_1$ by $v_1'(t_1) > 0$. 
     Note that the envy of agent 1 towards agent $j\neq 1 \in N$ is 
     \begin{align*}
         & v_1'(A_j) - v_1'(o_j) - v_1'(A_1) + v_1'(o_1) =\\
         &v_1(A_j) + v_1'(t_1) - v_1(o_j) - v_1(A_1) - v_1'(t_1) + v_1'(o_1) + v_1'(t_1).
     \end{align*}
     
     Since $t_1$ holds for every new apartment, increasing the value of $t_1$ results in increasing the values of all new apartments by the same amount, so the maximum-value assignment does not change. However,  the manipulation does affect the payments in favor of agent 1:
     \begin{align*}
         & p(v')_1 = \frac{n-1}{n^2}DP_1(A,v') - \frac{1}{n^2}\sum_{j\neq 1\in N} DP_j(A,v') = \\
         & \frac{n-1}{n^2}DP_1(A,v') - \frac{1}{n^2}\sum_{j\neq 1\in N} DP_j(A,v) = \\
         & \frac{n-1}{n^2}\left(DP_1(A,v) + \sum_{j\neq 1\in N} \left( v_1'(t_1) - v_1'(t_1) + v_1'(t_1) \right) \right) - \\
         &\frac{1}{n^2}\sum_{j\neq 1\in N} DP_j(A,v) > \\
         & \frac{n-1}{n^2}DP_1(A,v) - \frac{1}{n^2}\sum_{j\neq 1\in N} DP_j(A,v) = p(v)_1.
     \end{align*}
     Since the misreporting does not change the assignment, regardless of what the other agents report, agent 1 can still gain an advantage, demonstrating that the mechanism is safely manipulable.

\end{proof}

\underline{\bf Multiplicative Characteristics:}
We show a special case in which the above proof does not work.
Consider two agents, $ i_1 $ and $ i_2 $, and six characteristics, $ t_1, t_2, t_3, t_4, t_5, t_6 $. The values assigned by the agents to these characteristics are as follows:
\[
    \begin{bmatrix}
        & & t_1 & t_2 & t_3 & t_4 & t_5 & t_6\\
        &i_1 & 2 & 3 & 5 & 3 & 1 & 3 \\
        &i_2 & 1 & 2 & 2 & 2 & 4 & 8
    \end{bmatrix}
\]

The indicator variables for the apartments are defined as follows:
\[
    \begin{bmatrix}
        & & o_1 & o_2 & a_1 & a_2 \\
        &t_1 & 1 & 0 & 1 & 1 \\
        &t_2 & 0 & 1 & 0 & 0 \\
        &t_3 & 0 & 0 & 1 & 0 \\
        &t_4 & 0 & 0 & 0 & 1 \\
        &t_5 & 0 & 0 & 1 & 0 \\
        &t_6 & 0 & 0 & 0 & 1 
    \end{bmatrix}
\]
The base prices of the apartments are as follows:
\[
    \begin{bmatrix}
        & & \makecell{\text{Base Price}} \\
        & o_1 &  1 \\
        & o_2 &  1 \\
        & a_1 &  1 \\
        & a_2 &  1
    \end{bmatrix}
\]
That is, the total values for each apartment are as follows:
\[
    \begin{bmatrix}
        & & o_1 & o_2 & a_1 & a_2 \\
        &v_1 & 2 & 3 & 10 & 18 \\
        &v_2 & 1 & 2 & 8 & 16
    \end{bmatrix}
\]

The assignment where agent $i_1$ receives apartment $a_2$ and agent $i_2$ receives apartment $a_1$ yields the same total utility as the assignment where agent $i_1$ receives apartment $a_2$ and agent $i_2$ receives apartment $a_1$:
\[
v_1(a_1) + v_2(a_2) = 10 + 16 = 26 = 18 + 8 = v_1(a_2) + v_2(a_1).
\]

Without loss of generality, we can assume that agent $i_1$ receives apartment $a_2$ and agent $i_2$ receives apartment $a_1$ in this optimal assignment.

The payment agent $i_1$ pays under this assignment is:
\[
p(v)_1 = \frac{10 - 3 - 18 + 2}{4} - \frac{16 - 1 - 8 + 2}{4} = \frac{-9}{4} - \frac{9}{4} = -\frac{9}{2},
\]
and the utility of agent 1 is:
\[
u(v)_1 = v_1(a_2) + p(v)_1 = 18 - \frac{9}{2} = 13 + \frac{1}{2}.
\]

Assume agent 1 reports $v'(t_1) = x$ as the valuation of characteristic $ t_1 $. Then the new values become:
\[
    \begin{bmatrix}
        & & o_1 & o_2 & a_1 & a_2 \\
        &v_1' &x & 3 & 5x & 9x \\
        &v_2 & 1 & 2 & 8 & 16
    \end{bmatrix}
\]
There are two cases to consider:
\begin{enumerate}
    \item \textbf{The new assignment remains the same.}
    In this case,
    \begin{align*}
        & v_1'(a_1) + v_2(a_2) \leq v_1'(a_2) + v_2(a_1) \Leftrightarrow \\
        &
        5x + 16 \leq 9x + 8 \Leftrightarrow 8 \leq 4x \Leftrightarrow 2 \leq x.
    \end{align*}
    We say that, in order to keep the assignment unchanged, we need $2 < x$; otherwise agent $i_1$ does not change the valuation at all.
    
The payment agent 1 receives is:
\[
\begin{aligned}
    & p(v')_1 = \\
    &\frac{v_1'(a_1) - v_1'(o_2) - v_1'(a_2) + v_1'(o_1)}{4} - \\
    &\frac{v_2(a_2) - v_2(o_1) - v_2(a_1) + v_2(o_1)}{4} = \\
    &\frac{5x - 3 - 9x + x}{4} - \frac{9}{4} = \\
    &\frac{-3x - 12}{4} < \\
    &\frac{-6 -12 }{4} = \\
    &-\frac{9}{2}.
\end{aligned}
\]
Thus, the utility of agent 1 is:
\[
u(v')_1 = v_1(a_2) + p(v')_1 < 18 - \frac{9}{2} = u(v)_1.
\]

\item \textbf{The new assignment changes.} In this case, $x < 2$.

The payment agent 1 receives is:
\[
\begin{aligned}
    &p(v')_1 = \\
    &\frac{v_1'(a_2) - v_1'(o_2) - v_1'(a_1) + v_1'(o_1)}{4} - \\
    &\frac{v_2(a_1) - v_2(o_1) - v_2(a_2) + v_2(o_1)}{4} = \\
    &\frac{9x - 3 - 5x + x}{4} - \frac{8 - 1 - 16 + 2}{4} = \\
    &\frac{5x - 3}{4} - \frac{-7}{4} = \frac{5x + 4}{4} <\\
    &\frac{10 + 4}{4} = \\
    &\frac{7}{2} .
\end{aligned}
\]
Thus, the utility of agent 1 is:
\[
u(v')_1 = v_1(a_1) + p(v')_1 < 10 + \frac{7}{2} < 13 + \frac{1}{2} = u(v)_1.
\]
\end{enumerate}

In both cases, agent 1 risks decreasing his utility. Thus, the mechanism is Risk-Averse truthful.

\end{toappendix}


\begin{proposition}
With 
either  multiplicative or
additive characteristics,
when there are $n\geq 2$ agents and the
characteristics of  all new apartments are unknown,
the Minimum Disproportionality mechanism is risk-avoiding truthful.
\end{proposition}

\begin{proofsketch}
Recall that the Minimum Disproportionality mechanism (\Cref{def:min-envy-sum}) computes an assignment of the new apartments that maximizes the utilitarian welfare, and a balanced payment vector $\price{}$ by \eqref{eq:dp-payments}:
$
    p_i = DP_i(A) - \frac{DP_N(A)}{n}. 
$

Consider $n$ agents with old apartments $o_1, \ldots, o_n$, and suppose the new apartments are $a_1, \ldots, a_n$. 

Suppose agent 1 increases the values of some characteristics (in set $T_+ \subseteq T$), decreases the values of some other characteristics (in set $T_- \subseteq T\setminus T_+$), and does not change the values of the remaining characteristics (in set $T_0 := T\setminus (T_+ \cup T_-)$).

The characteristics of the $o_i$ are known to the manipulator, and do not affect the assignment. Hence, it is sufficient to check the case that  all characteristics in $T_+$ are in $T(o_1)$ and not in any $T(o_i)$ for $i\geq 2$,
and that all characteristics in $T_-$ are in every $T(o_i)$ for $i\geq 2$ and not in $T(o_1)$.
 This is clearly the best case for the manipulator, as it increases $DP_1(A)$ by the largest amount. We show that, even in that case, the manipulation is not safe. It follows that the manipulation is not safe in the cases less favorable for the manipulator.

Since the characteristics of all new apartments are unknown,  
it is possible that 
\begin{itemize}
\item 
some new apartment (say $a_1$) 
has all the characteristics in $T_+$ and no characteristic from $T_-$;
\item 
all other apartments $a_i$ for $i\ge 2$ have the same set of characteristics, which contains all the characteristics in $T_-$ and no characteristic from $T_+$; 
\item all other agents $i$ for $i\ge 2$ have the same valuation function, $v_2$.
\end{itemize}

We show that, in this case, it is possible that the assignment changes. Particularly, without the  manipulation agent 1 gets some $a_i$ for $i\geq 2$ (w.l.o.g. $a_2$, as all these apartments are identical); and with the manipulation, agent 1 gets $a_1$.
Note that, by the maximum-value matching definition, it suffices to show that $v_1(a_1) + v_2(a_2) < v_2(a_1) + v_1(a_2)$ and $v_1'(a_1) + v_2(a_2) > v_2(a_1) + v_1'(a_2)$, as $v_2$ is the valuation function of \emph{all} agents $2,\ldots,n$. These two conditions are equivalent to:
\begin{align}
	\label{eq:sufficient-condition-for-matching}
v_1(a_1) - v_1(a_2) < v_2(a_1) -  v_2(a_2) < v_1'(a_1) - v_1'(a_2).
\end{align}
We show a specific $v_2$ that satisfies these conditions. We also compute the changes in prices due to the manipulation, and show that, overall, agent 1 loses utility, which means that the manipulation is not safe.
\end{proofsketch}

\begin{proof}

Without manipulation, the payment received by agent 1 is (by \Cref{eq:dp-payments-explicit}):
\begin{align*}
	p(v)_1 = \frac{1}{n^2} \Bigg[p^A - p^O\Bigg],
\end{align*}
where $p^A$ and $p^O$ represent the contributions of new and old apartments respectively to the payment of agent 1 (for brevity we omit the subscript $1$). In this case $a_1$ is allocated to some agent $j\geq 2$; w.l.o.g. we assume it is agent 2, as all these agents are identical.
\begin{align*}
	p^A =& (n-1)\left[v_1(a_1) + \sum_{j\geq 3}v_1(a_j)\right]
	 - (n-1)^2 v_1(a_2)
	\\
	&- \left[\sum_{k\neq 1} v_2(a_k) +  \sum_{j\geq 3}\sum_{k\neq  j} v_j(a_k) \right]
	+ (n-1)\left[v_2(a_1) + \sum_{j\geq 3}v_j(a_j)\right]
	\\
	p^O &= (n-1)\sum_{j\geq 2}v_1(o_j)
	- (n-1)^2 v_1(o_1)
	- \sum_{j\geq 2}\sum_{k\neq  j} v_j(o_k) + \sum_{j\neq i}(n-1)v_j(o_j)
\end{align*}
The manipulation affects both $p^A$ and $p^O$ by changing $v_1$ to $v_1'$. In addition, it affects $p^A$ by changing the assignment of agent 1 to $a_1$ and of agent 2 to $a_2$. All in all, the price components with manipulation are:
\begin{align*}
	{p'}^A &= (n-1)\left[v_1'(a_2) + \sum_{j\geq 3}v_1'(a_j)\right]
	- (n-1)^2 v_1'(a_1)
	\\
	&- \left[\sum_{k\neq 2} v_2(a_k) - \sum_{j\geq 3}\sum_{k\neq  j} v_j(a_k) \right]
	+ (n-1)\left[v_2(a_2) + \sum_{j\geq 3}v_j(a_j) 	\right]
	\\
	{p'}^O &= (n-1)\sum_{j\geq 2}v_1'(o_j)
	- (n-1)^2 v_1'(o_1)
	- \sum_{j\geq 2}\sum_{k\neq  j} v_j(o_k) + \sum_{j\geq 2}(n-1)v_j(o_j)
\end{align*}
The difference in prices is $\frac{1}{n^2} \Bigg[\Delta^A - \Delta^O\Bigg]$, where
\begin{align*}
	\Delta^O := {p'}^O - p^O; && \Delta^A := {p'}^A - p^A.
\end{align*}
In addition to the price-change, agent 1 is also affected by the value change, $\Delta^V := v_1(a_1) - v_1(a_2)$.
Overall, the utility gain of agent 1 due to the manipulation is
$\frac{1}{n^2} \Bigg[\Delta^A - \Delta^O\Bigg] + \Delta^V
=
\frac{1}{n^2} \Bigg[\Delta^A  + n^2 \Delta^V - \Delta^O\Bigg]
$.
We aim to prove that the gain of agent 1 might be negative, that is, we should prove that $\Delta^A  + n^2 \Delta^V - \Delta^O < 0$.

We now compute simplified expressions for each component of this expression. We first simplify $\Delta^O$:
\begin{align}
	\notag
	\Delta^O &= (n-1)\sum_{j\geq 2}[v_1'(o_j) - v_1(o_j)]
- (n-1)^2 [v_1'(o_1) - v_1(o_1)]
\notag
\end{align}
To simplify this expression, we define $v_1(o_*)$ as the average value (for agent 1) of the $n-1$ apartments of the other agents, and similarly for $v_1'(o_*)$:
\begin{align*}
	v_1(o_*) := \frac{1}{n-1}\sum_{j\geq 2}v_1(o_j);
	&&
	v_1'(o_*) := \frac{1}{n-1}\sum_{j\geq 2}v_1'(o_j)
\end{align*}
Hence:
\begin{align}
	\label{eq:DeltaO}
	\Delta^O &= (n-1)^2[(v_1'(o_*) - v_1(o_*)) - (v_1'(o_1) - v_1(o_1))].
\end{align}

Next, we consider $\Delta^A$:
\begin{align}
\notag
	\Delta^A &= 
	(n-1)\left[v_1'(a_2)-v_1(a_1) + \sum_{j\geq 3}(v_1'(a_j)-v_1(a_j))\right]
	\\
\notag
	&- (n-1)^2 [v_1'(a_1) - v_1(a_2)]
	\\
\notag
&- [v_2(a_1) - v_2(a_2)] 
\\
\notag
&+ (n-1)\left[v_2(a_2)  - v_2(a_1)\right].
\end{align}
Note the last two lines can be simplified to $n\cdot [v_2(a_2)-v_2(a_1)]$.
Hence,
\begin{align}
\notag
	\Delta^A + n^2\Delta^V &= 	
	(n-1)\left[v_1'(a_2)-v_1(a_2) + \sum_{j\geq 3}(v_1'(a_j)-v_1(a_j))\right]
\\
\notag
&- (n-1)^2 [v_1'(a_1) - v_1(a_1)]
+ n[v_2(a_2)  - v_2(a_1)]
	+ n\cdot \Delta^V.
\notag
\end{align}
This expression can be further simplified, as all apartments $a_2,\ldots,a_n$ are identical:
\begin{align}
	\notag
	\Delta^A + n^2\Delta^V &= 	
	(n-1)^2\left[(v_1'(a_2)-v_1(a_2)) - (v_1'(a_1) - v_1(a_1))\right]
	\\
	&+ n[(v_2(a_2)  - v_2(a_1)) + (v_1(a_1)-v_1(a_2))].
	\label{eq:DeltaAV}
\end{align}
Note that, by inequality \eqref{eq:sufficient-condition-for-matching} (the sufficient condition for matching), the last line of the above expression is negative. 
Hence, to prove that $\Delta^A  + n^2 \Delta^V - \Delta^O < 0$, it is sufficient to prove, in addition to \eqref{eq:sufficient-condition-for-matching}, that the following holds:
\begin{align}
	\notag
	&(n-1)^2\left[(v_1'(a_2)-v_1(a_2)) - (v_1'(a_1) - v_1(a_1))\right]
	\\
	\label{eq:sufficient-condition-for-loss}
	-&
	(n-1)^2[(v_1'(o_*) - v_1(o_*)) - (v_1'(o_1) - v_1(o_1))]
	\leq 0.
\end{align}

We now prove conditions \eqref{eq:sufficient-condition-for-matching} and \eqref{eq:sufficient-condition-for-matching} separately for additive and multiplicative characteristics.

\underline{\bf Additive Characteristics:} 
By assumption, for all $t\in T_+$ there is $z_t > 0$ such that $v'_1(t) = v_1(t) + z_t$ and 
for all $t\in T_-$ there is a 
$y_t < 0$ such that $v'_1(t) = v_1(t) + y_t$, and 
and $v'_1(t) = v_1(t)$ 
for all other $t$.
Denote 
\begin{align*}
z := \sum_{t\in T_+} z_t; && y := \sum_{t\in T_-} y_t.
\end{align*}

Since $T_+\subseteq T(o_1)\setminus T(o_j)$ and 
$T_-\subseteq T(o_j)\setminus T(o_1)$ for all $j\geq 2$,
\begin{align*}
v_1'(o_1) = v_1(o_1) + z > v_1(o_1); && 
v_1'(o_j) = v_1(o_j) + y < v_1(o_j).
\end{align*}

Since $T_+\subseteq T(a_1)\setminus T(a_2)$ and 
$T_-\subseteq T(a_2)\setminus T(a_1)$,
\begin{align*}
v_1'(a_2) = v_1(a_2) + y < v_1(a_2); && 
v_1'(a_1) = v_1(a_1) + z > v_1(a_1).
\end{align*}

We now construct the valuation function $v_2$ of agents $2,\ldots,n$  such that conditions \eqref{eq:sufficient-condition-for-matching} hold, that is,
\begin{align}
	\label{eq:sufficient-condition-additive}
	&
	v_1(a_1) - v_1(a_2) < v_2(a_1) -  v_2(a_2) <  v_1(a_1) - v_1(a_2) + z - y;
\end{align}
this should be possible as $z > 0 > y$.
Specifically, suppose that 
\begin{align*}
v_2(t)  = 
\begin{cases}
    v_1(t) + \gamma_t ~~~ (\text{ where } z_t>\gamma_t>0) & t\in T_+;
    \\
    v_1(t) + \delta_t ~~~ (\text{ where } y_t<\delta_t<0) & t\in T_-;
    \\
    v_1(t) & \text{otherwise}
\end{cases}
\end{align*}

Denote 
\begin{align*}
	\gamma := \sum_{t\in T_+} \gamma_t; && \delta := \sum_{t\in T_-} \delta_t.
\end{align*}

Since $a_1$ contains all characteristics from $T_+$ and none from 
$T_-$,
\begin{align*}
    & v_1(a_1) - \sum_{t\in T_+}v_1(t) = 
    v_2(a_1) - \sum_{t\in T_+}v_2(t) =\\
                      & v_2(a_1) - \sum_{t\in T_+}v_1(t) - \sum_{t\in T_+}\gamma_t  ;
                      \\
                     \Rightarrow &
v_1(a_1) = v_2(a_1) - \gamma.
\end{align*}
Similarly,  $a_2$ contains all characteristics from $T_-$ and none from $T_+$,
\begin{align*}
    & v_1(a_2) - \sum_{t\in T_-}v_1(t) = 
    v_2(a_2) - \sum_{t\in T_-}v_2(t) =\\
                      & v_2(a_2) - \sum_{t\in T_-}v_1(t) - \sum_{t\in T_-}\delta_t  ;
                      \\
                      \Rightarrow &
v_1(a_2)  = v_2(a_2)- \delta.
\end{align*}
Thus,
\begin{align*}
v_2(a_1) - v_2(a_2)  = v_1(a_1) - v_1(a_2) + \gamma-\delta.
\end{align*}
This satisfies \eqref{eq:sufficient-condition-additive} as $y<\delta<0<\gamma<z$, so $0 <\gamma-\delta < z-y$.
Thus, without manipulation agent 1 gets $a_2$ (or any equivalent apartment $a_j$ with $j\ge 2$), and with manipulation agent 1 gets $a_1$.

It remains to prove inequality \eqref{eq:sufficient-condition-for-loss}.
Indeed,
$
v_1'(a_2) = v_1(a_2)+y, 
v_1'(a_1) = v_1(a_1)+z, 
v_1'(o_*) = v_1(o_*)+y, 
v_1'(o_1) = v_1(o_1)+z$,
so the difference is $(n-1)^2[(y-z)  - (y-z)] = 0$.
Combining with \eqref{eq:sufficient-condition-for-matching} implies that the utility gain of agent 1 due to the manipulation is negative.

\underline{\bf Multiplicative Characteristics:} 
For all $t \in T$ there is a factor $z_t > 1$ such that $v'_1(t) = z_t \cdot v_1(t)$ for all $t\in T_+$,
and a factor $0 < y_t<1$ for 
such that $v'_1(t) = y_t \cdot v_1(t)$
all $t\in T_-$,
and $v'_1(t)=v_1(t)$ for all other $t \in T$.
~
Denote 
\begin{align*}
z := \prod_{t\in T_+} z_t; && y := \prod_{t\in T_-} y_t.
\end{align*}

Since $T_+\subseteq T(o_1)\setminus T(o_j)$ and 
$T_-\subseteq T(o_j)\setminus T(o_1)$ for all $j\geq 2$,
\begin{align*}
v_1'(o_1) = z\cdot  v_1(o_1) > v_1(o_1); && 
v_1'(o_j) = y\cdot v_1(o_j) < v_1(o_j).
\end{align*}

Since $T_+\subseteq T(a_1)\setminus T(a_2)$ and 
$T_-\subseteq T(a_2)\setminus T(a_1)$,
\begin{align*}
v_1'(a_1) = z \cdot v_1(a_1) > v_1(a_1); && 
v_1'(a_2) = y \cdot v_1(a_2) < v_1(a_2).
\end{align*}

We construct $v_2$ such that conditions \eqref{eq:sufficient-condition-for-matching} hold, that is,
\begin{align}
\label{eq:sufficient-condition-multiplicative}
	v_1(a_1) - v_1(a_2) < v_2(a_1) -  v_2(a_2) < z v_1(a_1) - y v_1(a_2).
\end{align}
Similarly to the additive case, this can be ensured by defining $v_2$ as follows:
\begin{align*}
	v_2(t)  = 
	\begin{cases}
		v_1(t) \cdot \gamma_t ~~~ (\text{ where } z_t>\gamma_t>1) & t\in T_+;
		\\
		v_1(t)\cdot  \delta_t ~~~ (\text{ where } y_t<\delta_t<1) & t\in T_-;
		\\
		v_1(t) & \text{otherwise}
	\end{cases}
\end{align*}

It remains to prove inequality \eqref{eq:sufficient-condition-for-loss}.
Here,
$
v_1'(a_2) = y\cdot v_1(a_2), 
v_1'(a_1) = z\cdot v_1(a_1), 
v_1'(o_*) = y\cdot v_1(o_*), 
v_1'(o_1) = z\cdot v_1(o_1)$,
so the difference is 
\begin{align*}
	(n-1)^2 [ (y-1)[v_1(a_2) - v_1(o_*)] - (z-1)[v_1(a_1)-v_1(o_1)]].
\end{align*}
This expression is indeed at most $0$ if the following two additional conditions hold for the new apartments:
\begin{enumerate}
	\item $v_1(a_1) \geq v_1(o_1)$;
	\item $v_1(a_2) \geq v_1(o_*)$;
\end{enumerate}
For the first condition, note that both $a_1$ and $o_1$ contain all characteristics from $T_+$ and no characteristic from $T_-$. It is also possible that $a_1$ contains all characteristics in $T_0$ with value larger than 1, and no characteristic in $T_0$ with value smaller than 1. Hence, $v_1(a_1) \geq v_1(o_1)$.

Similarly, for the second condition, note that both $a_2$ and $o_j$ for $j\geq 2$ contain all characteristics from $T_-$ and no characteristic from $T_+$. It is also possible that $a_2$ contains all characteristics in $T_0$ with value larger than 1, and no characteristic in $T_0$ with value smaller than 1. Hence, $v_1(a_2) \geq v_1(o_*)$, so the difference is at least $0$.

Combining with \eqref{eq:sufficient-condition-for-matching} implies that the utility gain of agent 1 is negative.
\end{proof}

\section{Experiments} \label{sec: exp}
The possible non-existence of envy-free and proportional allocations has motivated us to check what envy can be attained in realistic Reconstruct and Divide projects.

\subsection{Apartment characteristics}
Initially, we constructed a list $T$ of 18 apartment characteristics, based on 
characteristics used by appraisers in Reconstruct and Divide projects, as well as characteristics used in Yad2 --- an Israeli website for selling used apartments (similar to Zillow).

These characteristics capture both structural and qualitative aspects of housing value, such as renovation, balcony, garden apartment, registered parking, storage, directions of exposure, floor level, elevator, accessibility, room count, spaciousness, building rights, orientation, road exposure, air directions, and natural light \footnote{
The final set of characteristics includes: renovated apartment, balcony, garden apartment, registered parking, storage, safe room, directions of exposure, high floor, low floor, elevator, accessibility, more than four rooms, fewer than four rooms, spacious living area or kitchen, building rights, view-facing apartment, road-facing apartment, three or more air directions, and natural light.}.

In addition, we collected empirical data from 28 apartments involved in real Reconstruct and Divide projects in Jerusalem and Haifa, Israel. For each apartment (both pre-renewal and newly constructed units), we extracted its size (in square meters), its price per square meter, and the presence or absence of each characteristic in $T$, based on transaction databases and project documentation.

If information about a characteristic was missing for a given apartment, we completed it by assigning a value of 0 or 1 independently with probability $1/2$. The base price of each apartment is computed as its size multiplied by its price per square meter.

\subsection{Valuation elicitation}
We assumed that agents' utilities are generated by multiplicative characteristics (see \Cref{sec: manipulations}), which is more similar to the method used in practice by real-estate appraisers.

We conducted a survey among $45$ apartment owners (provided to us by a commercial Internet-panel company). Each participant was asked to specify the percentage by which they believe each characteristic (positively or negatively) affects an apartment’s base price. If respondent $i$ assigns a positive influence of $X$ percent to a characteristic $t \in T$, we define $v_i(t) = 1+\frac{X}{100}$. If the influence is negative, we define $v_i(t) = \max\{1-\frac{X}{100}, 10^{-5}\}$.

Designing this elicitation method was nontrivial and required several iterations and pilot studies. In early versions, participants were asked to evaluate characteristics in isolation (e.g., “How much does a balcony add to the value of an apartment?”). This led to unrealistically large reported effects—often exceeding $90\%$, corresponding to multiplicative factors above $1.9$—which, when combined across characteristics, produced implausibly high valuations.

Based on these pilot results, we redesigned the survey to elicit valuations in context. Instead of evaluating characteristics in isolation, participants were asked to assess the marginal impact of adding a characteristic to a concrete apartment description (e.g., “If an apartment without a balcony is worth 1 million, how much would you pay for the same apartment with an added balcony?”). This contextualized elicitation approach yielded substantially more realistic and stable valuations. All data used in our experiments were collected using this final survey design.

In our initial experiment, we observed that some respondents assigned near-100\% positive influence to more than ten characteristics, leading to large multiplicative effects on the base apartment value and, in some cases, to apartment valuations that exceeded 10 million. Since such values lie well beyond a realistic market scale, we introduced a normalization procedure that is applied only when extreme values occur.

Rather than excluding these responses, we retained them and applied a structure-preserving normalization triggered by a threshold on the cumulative impact of characteristics. 

Formally, for each agent $i$ and apartment $a$ in the dataset, we compute the product of the agent's values for the characteristics of that apartment, $w_i^a := \prod_{t\in T}\rho_{i,t}^{\beta_{t,a}}$. We then compute for each agent $i$ the maximum of that product over all apartments, $w_i^{\max} := \max_{a}w_i^a$.
If $w_i^{\max}$ is larger than a threshold $W$, we normalize all characteristic values for $i$ logarithmically as follows, so that the maximum product for $i$ is exactly $W$:
\begin{align*}
    \rho_{i,t}^* := \exp\left(
        \log(\rho_{i,t})\cdot \frac{\log(W)}{\log(w_i^{\max})}
    \right)
\end{align*}
If $w_i^{\min} := \min_{a}w_i^a$ is smaller than $1/W$, we normalize in a similar way.
In the experiments we used $W=2$ for the normalization.

\subsection{Experiment}
We generate an instance of Reconstruct and Divide as follows.

We generate $45$ agents, with valuations taken from the $45$ responses to our survey.
We assigned each agent $i$ a random old apartment $o_i$ and computed their valuations for all new apartments using a multiplicative utility function (as commonly employed by real-estate appraisers).
For example, if one subject gave valuations of $1.1$, $1.2$ and $0.9$ to the three characteristics present in the apartment, then we computed this subject's value for the entire apartment as $\text{base-value} \cdot 1.1\cdot 1.2 \cdot 0.9$.


To capture behavioral biases observed in real housing markets, we examined the model’s sensitivity to the \emph{endowment effect} --- the tendency of individuals to assign a different value to goods they already own than to identical goods they do not own \cite{wikipedia-Endowment-effect}. 

For each characteristic $t\in T$, participants indicated whether their current apartment possessed that characteristic, and we then estimated, for each $t$, a separate linear regression in which the coefficient represents the ownership premium --- the difference between the value assigned to a characteristic when it is part of the respondent’s own apartment versus when it appears in alternative apartments. 
\Cref{tab:endowment-effects} reports the estimated coefficients and their corresponding $p$-values ($Prob>|t|$).

\begin{table}[H]
\caption{Characteristics, endowment-effect coefficients, and $p$-values.}
\label{tab:endowment-effects}
\centering
\setlength{\tabcolsep}{6pt}
\begin{tabularx}{\textwidth}{|X|X|X|}
\hline
\textbf{Characteristic} & \textbf{Endowment effect} & \textbf{$p$-value} \\
\hline
Renovated apartment        & -15.57018  & 0.0840  \\ \hline
Balcony                    & -7.406832  & 0.1281  \\ \hline
Garden apartment           & 10.00000   & 0.1639  \\ \hline
Registered parking         & -6.18254   & 0.1995  \\ \hline
Storage                    & 1.92547    & 0.6444  \\ \hline
Safe Room                  & -7.27941   & 0.2720  \\ \hline
Air conditioning           & -9.92424   & 0.3789  \\ \hline
High floor                 & -10.93750  & 0.0440  \\ \hline
Low floor                  & -6.53509   & 0.1359  \\ \hline
Elevator                   & 0.17081    & 0.9742  \\ \hline
Accessible                 & 0.27941    & 0.9579  \\ \hline
More than 4 rooms          & -13.36765  & 0.1298  \\ \hline
Less than 4 rooms          & -4.30000   & 0.4628  \\ \hline
Spacious living/kitchen    & -10.02381  & 0.2170  \\ \hline
Building rights            & -11.25926  & 0.2557  \\ \hline
View-facing apartment      & 12.65839   & 0.0266  \\ \hline
Road-facing apartment      & -6.11888   & 0.5410  \\ \hline
3+ air directions          & 1.63235    & 0.7796  \\ \hline
Natural light              & 14.36275   & 0.1536  \\ \hline
\end{tabularx}
\end{table}

Let $Val_{i,t}$ denote the value reported by agent $i$ for characteristic $t$.  
The adjusted valuation used in the experiments was defined as
\[
\tilde v_{i,t} =
\begin{cases}
Val_{i,t}, & \text{if the ownership status of } t \text{ is unchanged},\\[1mm]
Val_{i,t}+\beta_t, & \text{if } t \text{ $i$ owns $t$ in experiment but not in survey},\\[1mm]
Val_{i,t}-\beta_t, & \text{if } t \text{ $i$ owns $t$ in survey but not in experiment}.
\end{cases}
\]

For example, for \emph{renovated apartment} characteristic, with an endowment effect of $-15.57\%$, gaining the characteristic decreases the valuation by $15.57\%$, whereas losing it increases the valuation by the same amount.

Only a small number of characteristics exhibited statistically significant ownership effects ($p<0.05$). Hence, our current results do not provide evidence for an endowment effect. On the other hand, the inconclusive results might be due to the small sample size.
Hence, we decided to conduct two separate experiments: one using the raw values and ignoring the endowment effect, and one using the adjusted values.
We hope to further investigate the strength and consistency of the endowment effect in future work.

After computing the valuations, we compute an assignment $A$ maximizing utilitarian welfare, by computing a maximum-weight bipartite matching between the agents and their valuations for the new apartments (via the Hungarian algorithm).


To find a payment vector that minimizes the largest envy, we solve the following linear program:
\begin{align*}  
    \displaystyle\min_{\price{},z}& z \\
    \text{s.t. } \forall i\neq j: (v_i(a_j) - v_i(o_j) + p_j) &- (v_i(a_i) - v_i(o_i) + p_i) \leq z \\
    \sum_{i} p_i &= 0
\end{align*}

To find the payment vector that minimizes the largest disproportionality, we apply the payment rule defined by the Minimum Disproportionality mechanism (see \Cref{def:min-envy-sum}).

Finally, for both payment vectors, we compute the resulting disproportionality and envy experienced by any agent.
Moreover, we calculate the average valuation across all new apartments and agents, denoted by $V$. To obtain a more meaningful metric, we normalize the disproportionality and envy values with respect to this average. 
The results for the experiments with raw values and with the endowment effect are presented in \Cref{tab:experiment_without_endowment} and \Cref{tab:experiment_with_endowment} respectively.
All reported results are averaged over 10 independent runs.


\begin{table}[H]
\centering
    \caption{Comparison of envy and disproportionality, evaluated under identical allocations of old and new apartments without endowment-effect correction. Each experiment is repeated 10 independent runs, and the values reported in the table are averages across these runs. The normalization threshold used is $W=2$ and the average valuation of the new apartments is $V\approx 6,194,371$ ILS.\label{tab:experiment_without_endowment}}
\renewcommand{\arraystretch}{1.05}
\setlength{\tabcolsep}{6pt}

\begin{tabularx}{\textwidth}{|X|X|X|X|X|}
\hline
& \makecell{\textbf{Largest} \\ \textbf{dispropor-}\\\textbf{tionality}} & $\frac{\makecell{\textbf{Largest} \\ \textbf{dispropor-}\\\textbf{tionality}}}{V}$ & \textbf{Largest envy} &$\frac{\textbf{Largest envy}}{V}$\\
         \midrule
         \makecell{\textbf{Minimizing} \\\textbf{Envy} \\ \textbf{Payments}} & 630051.71 & 0.102 & \textbf{1256569.381} & \textbf{0.203}\\
         \midrule
         \makecell{\textbf{Minimizing} \\ \textbf{Dispropor-}\\\textbf{tionality} \\ \textbf{Payments}} &  \textbf{-753708.779} & \textbf{-0.121} & 12317415.16 & 1.986\\
         \bottomrule
\end{tabularx}
\end{table}

\begin{table}[H]
\centering
    \caption{Comparison of envy and disproportionality, evaluated under identical allocations of old and new apartments with endowment-effect correction. Each experiment is repeated 10 independent runs, and the values reported in the table are averages across these runs. The normalization threshold used is $W=2$ and the average valuation of the new apartments is $V\approx 6,178,831$ ILS.\\ \label{tab:experiment_with_endowment}}
\renewcommand{\arraystretch}{1.05}
\setlength{\tabcolsep}{6pt}

\begin{tabularx}{\textwidth}{|X|X|X|X|X|}
\hline
& \makecell{\textbf{Largest} \\ \textbf{dispropor-}\\\textbf{tionality}} & $\frac{\makecell{\textbf{Largest} \\ \textbf{dispropor-}\\\textbf{tionality}}}{V}$ & \textbf{Largest envy} &$\frac{\textbf{Largest envy}}{V}$\\
         \midrule
         \makecell{\textbf{Minimizing-}\\\textbf{Envy} \\ \textbf{Payments}} & 470127.155 & 0.076 & \textbf{1309549.262} & \textbf{0.212}\\
         \midrule
         \makecell{\textbf{Minimizing} \\ \textbf{Dispropor-}\\\textbf{tionality} \\ \textbf{Payments}} &  \textbf{-1196156.687} & \textbf{-0.193} & 20023748.698 & 3.235\\
         \bottomrule
\end{tabularx}
\end{table}

\ref{tab:experiment_without_endowment},\ref{tab:experiment_with_endowment}, we observe that, although we could not find an EF-able allocation---i.e., the maximum envy remains positive---it is still possible to achieve a PROP-able allocation, where the disproportionality with payments is negative.
Moreover, as expected, a negative disproportionality (corresponding to a proportional allocation) does not necessarily imply that the maximum envy is negative too. 
At the same time, these results highlight that a simple utilitarian-maximization algorithm—originally developed for computing envy-freeable allocations—can be effectively repurposed to minimize disproportionality in the more complex setting with endowments.

Moreover, after normalization, we find that the maximum envy reaches approximately 50\% of an apartment's value--—indicating a high level of envy. This strengthens our motivation to address \Cref{open question: macc}, whether a better assignment exists that can achieve a lower maximum envy.

\section{Conclusions and Future Work}
We studied fairness in Reconstruct and Divide projects by considering two common fairness notions: envy-freeness and proportionality. We characterized when envy-free and proportional assignments exist. Since such assignments are not always guaranteed, we focused on minimizing envy and disproportionality through assignment and payment vectors.

We showed that the maximum attainable envy equals the maximum average mean cycle cost in the difference envy-graph. However, 
we could not find a polynomial-time algorithm that guarantees an optimal solution. Whether such an algorithm exists remains an open question.

On the positive side, we proposed a mechanism that minimizes maximum disproportionality. We also introduced a characteristic-based elicitation method to reduce manipulation. We also identified conditions under which the mechanism is resistant to safe manipulations. Extending these results to broader settings and identifying additional manipulation-resistant conditions remains an open direction.

We assumed that improvements are measured by differences; the model where improvements are measured by \textbf{ratios} is also conceptually appealing, but our current results for it are limited. Developing deeper theoretical insights in this model is a key area for future research.

In future work, we plan to expand our survey to include a larger and more diverse group of participants 
Our findings also highlight the importance of developing improved allocation methods that further reduce envy. In addition, we aim to gain a deeper understanding of the behavioral endowment effect through follow-up experiments designed to more accurately quantify its influence on perceived fairness and allocation outcomes.

\begin{toappendix}
    \section{Valuations as a Set of Characteristics in the maximum envy model} \label{app: Valuations as a Set of Characteristics}
\begin{proposition} \label{collection of characteristics no EF-able}
    There is a collection of characteristics, original apartments $O$, and a set of new apartments $M$ such that, for every assignments of one apartment per agent, the resulting assignment is not EF-able.
\end{proposition}
\begin{proof}
    We consider 4 characteristics $t_1, t_2, t_3, t_4$.

    The values assigned by the agents to these characteristics are as follows:
    \[
    \begin{bmatrix}
          & t_1 & t_2      & t_3        & t_4       \\
      i_1 & z   & \epsilon  & 0         &  0        \\
      i_2 & y   & 0         & \epsilon  &  0        \\
      i_3 & x   & 0         & 0         & \epsilon
    \end{bmatrix}.
    \]
    where $x,y,z,\epsilon \in \mathbb{N}$.
    The indicator variables for the original apartments are defined as follows:
    \[
        \begin{bmatrix}
              & t_1 & t_2 & t_3 & t_4  \\
          o_1 & 1   & 1   & 0   & 0 \\
          o_2 & 1   & 0   & 1   & 0 \\
          o_3 & 1   & 0   & 0   & 1 \\
        \end{bmatrix}.
    \]

    The indicator variables for the new apartments are defined as follows:
    \[
        \begin{bmatrix}
              & t_1 & t_2 & t_3 & t_4  \\
          a_1 & 1   & 1   & 1   & 1 \\
          a_2 & 1   & 1   & 1   & 1 \\
          a_3 & 1   & 1   & 1   & 1 \\
        \end{bmatrix}.
    \]

    Thus, the total values of the old apartments are:
    \[
    \begin{bmatrix}
            & o_1           & o_2           & o_3       \\
        i_1 & z+\epsilon    & z             & z         \\
        i_2 & y             & y+ \epsilon   & y         \\
        i_3 & x             & x             & x + \epsilon
    \end{bmatrix}
    \]
    and the total values of the new apartments are:
    \[
    \begin{bmatrix}
            & a_1           & a_2           & a_3       \\
        i_1 & z+\epsilon    & z+\epsilon    & z+\epsilon\\
        i_2 & y+\epsilon    & y+\epsilon    & y+\epsilon\\
        i_3 & x+\epsilon    & x+\epsilon    & x+\epsilon
    \end{bmatrix}
    \]

Let $C$ be the cycle $i_1 \to  i_2 \to i_1$.
Then $cost_O(C) = -3\epsilon$.
In contrast, for any assignment $A$ of the new apartments, $cost_A(C) = 0$.
Hence, $cost_{A,O}(C) = cost_A(C) - cost_O(C) > 0$. By \Cref{theorem:mefable-iff-cost-cycles}, $A$ is not EF-able.
\end{proof}

\section{Envy-Freeable: Ratio Model} \label{section: Ratio Model}
\subsection{The Model}
%
In this model, each agent evaluates what is the fixed entitlement of every other agent. That is, the entitlement of agent $j\in N$, in agent $i\in N$ point of view is $w_{i,j} = \frac{v_i(o_j)}{\sum_{k\in N} v_i(o_k)} \geq 0$. 
Note that for each $i \in N$, it holds that $\sum_{j\in N} w_{i,j} = 1$.

\begin{Remark}
    This model differs from the one presented in \cite{elmalem2025whoever}. Here, the values of the original apartments are subjective and may vary between agents. In contrast, if the values were objective—such as when determined by an appraiser—Elmalem et al.'s \cite{elmalem2025whoever} model would be applicable.
\end{Remark}

\begin{Definition}
The \textit{ratio-envy} of agent $i$ towards agent $j$ under an assignment $A$
and payments $\price{}$ is defined as 
\[\frac{u_i(A_j)}{w_{i,j}} - \frac{u_i(A_i)}{w_{i,i}}
= \frac{v_i(A_j)+p_j}{w_{i,j}} - \frac{v_i(A_i)+p_i}{w_{i,i}}.\]
assignment $A$ and payment vector $\price{}$ are \emph{Ratio-envy-free (REF)} if 
for each $i,j \in N$:
\[
\frac{v_i(A_j)+p_j}{w_{i,j}} - \frac{v_i(A_i)+p_i}{w_{i,i}} \leq 0
\]
\end{Definition}
\begin{Definition}
An assignment $A$ is \emph{ratio envy-freeable (REF-able)} if there exists a vector of payments $\price{}$ such that $(A,\price{})$ are REF.
\end{Definition}

\begin{Definition}
    The \emph{ratio envy graph} of assignment $A$, denoted $G_{A,w}$, is a complete directed graph consisting of a set of vertices representing agents $N$.

For any pair of agents $i$ and $j$ in $N$, the cost assigned to the arc $(i,j)$ in $G_{A,w}$ is defined as the ratio envy that agent $i$ holds toward agent $j$ under assignment $A$: 
$cost_A(i,j) := \frac{v_i(A_j)}{w_{i,j}} - \frac{v_i(A_i)}{w_{i,i}}$.
We denote the cost of a path $(i_1,...,i_k)$ as $cost_A(i_1,...,i_k) = \sum_{j = 1}^{k-1} cost_A(i_j, i_{j+1})$.
\end{Definition}

\subsection{REF assignment}

The following proposition demonstrates that there exist scenarios in which no assignment and payment vector satisfy the REF criteria:

\begin{proposition} \label{prop: not REF-able}
    There are original apartments $O$ and a set of new apartments $M$ such that, for every assignment of one apartment per agent, the resulting assignment is not REF-able.
\end{proposition}
\begin{proof}
Consider the following example with 2 agents.
The values of the original apartments are
\[
\begin{bmatrix}
   & o_1 & o_2 \\
  i_1 & z + \epsilon & z \\
  i_2 & y & y + \epsilon\\
\end{bmatrix}
\]
The values for the new apartments are
\[
\begin{bmatrix}
   & a_1 & a_2 \\
  i_1 & V & V\\
  i_2 & V & V \\
\end{bmatrix}
\]
when $x,y,\epsilon,V \in \mathbb{N}$.

Notice that $w_{1,1} = \frac{z+\epsilon}{2z+\epsilon}, w_{1,2}=\frac{z}{2z+\epsilon}$ and $w_{2,2} = \frac{y+\epsilon}{2y+\epsilon}, w_{2,1}=\frac{y}{2y+\epsilon}$.
    
For any assignment $A$ such that each agent receives exactly one item, assuming there is a balanced payment vector $\price{}$, we have:
\begin{align*}
& \begin{cases}
    \frac{v_1(A_2) + p_2}{w_{1,2}} \leq \frac{v_1(A_1) + p_1}{w_{1,1}} \\
    \frac{v_2(A_1) + p_1}{w_{2,1}} \leq \frac{v_2(A_2) + p_2}{w_{2,2}}
\end{cases}
\Leftrightarrow \\
&
 \begin{cases}
    (V + p_2)\frac{2z+\epsilon}{z} \leq (V + p_1)\frac{2z+\epsilon}{z+\epsilon} \\
    (V+ p_1)\frac{2y+\epsilon}{y} \leq (V + p_2)\frac{2y+\epsilon}{y+\epsilon}
\end{cases}
\Leftrightarrow^{p_1 = -p_2} \\
&\begin{cases}
    (V + p_2)\frac{2z+\epsilon}{z} \leq (V - p_2)\frac{2z+\epsilon}{z+\epsilon} \\
    (V - p_2)\frac{2y+\epsilon}{y} \leq (V + p_2)\frac{2y+\epsilon}{y+\epsilon}
\end{cases}
\Leftrightarrow \\
&\begin{cases}
    (V + p_2)\cdot (z+\epsilon) \leq (V - p_2)\cdot z \\
    (V - p_2) \cdot (y+\epsilon) \leq (V + p_2)\cdot y
\end{cases}
\Leftrightarrow 
\\
&\begin{cases}
    Vz + zp_2 + V\epsilon + \epsilon p_2 \leq Vz - zp_2 \\
    Vy - yp_2 + V\epsilon - \epsilon p_2 \leq Vy - yp_2
\end{cases}
\Leftrightarrow \\
&\begin{cases}
     2zp_2  + \epsilon p_2 \leq -V\epsilon\\
     2yp_2 + \epsilon p_2\geq V\epsilon
\end{cases} \Leftrightarrow 
\\
&\begin{cases}
     p_2 \leq -\frac{V\epsilon}{2z + \epsilon}\\
    p_2 \geq \frac{V\epsilon}{2y + \epsilon}
\end{cases}
\end{align*}
If $z\geq y$ then $\frac{V\epsilon}{2y + \epsilon} \geq \frac{V\epsilon}{2z + \epsilon}$. Hence $p_2 \leq -\frac{V\epsilon}{2z + \epsilon}$ and $p_2 \geq \frac{V\epsilon}{2z + \epsilon}$. Obviously, as all variables are positive, there is no solution.

If $y\geq z$ then $-\frac{V\epsilon}{2z + \epsilon} \leq -\frac{V\epsilon}{2y + \epsilon}$. Hence $p_2 \leq -\frac{V\epsilon}{2y + \epsilon}$ and $p_2 \geq \frac{V\epsilon}{2y + \epsilon}$. Obviously, as all variables are positive, there is no solution.
\end{proof}

\end{toappendix}
\appendix

\section*{Acknowledgments}
We want to thank Daniel Halpern for a fruitful discussion that helped obtain the results in \Cref{sec: max mean cycle cost},
and Shaul Tzionit for his help in the statistical analysis of the experiment results.
Erel Segal-Halevi is funded by Israel Science Foundation grants no. 712/20 and 1092/24.
Rica Gonen and Noga Klein Elmalem are funded by Environment and Sustainability Research Center grant 7/24.
\newpage
\bibliographystyle{ACM-Reference-Format}
\bibliography{ref}

\newpage
\end{document}